\documentclass[12pt, draftcls, peerreview, onecolumn]{IEEEtran}
\usepackage{amsmath}
\usepackage{amsthm}
\usepackage{amsfonts}
\usepackage{amssymb}
\usepackage{amscd}
\usepackage{xcolor}
\usepackage{soul}
\usepackage{cancel}
\usepackage{graphicx}
\usepackage{newlfont}
\usepackage{cite}
\usepackage{color}
\usepackage{multirow,tabularx}
\usepackage{ifthen}
\usepackage{enumitem}
\usepackage{times}
\usepackage{stackrel}
\usepackage[ruled]{algorithm2e}
\usepackage{algpseudocode}
\usepackage{float}
\usepackage[all]{xy}
\usepackage{url}
\usepackage{tabularx}
\usepackage{booktabs}
\usepackage{multirow}
\usepackage{array}
\usepackage{tablefootnote}

\setstcolor{red}

\setlist[description]{font=\normalfont\itshape\space}

\def\BibTeX{{\rm B\kern-.05em{\sc i\kern-.025em b}\kern-.08em
    T\kern-.1667em\lower.7ex\hbox{E}\kern-.125emX}}

\IEEEoverridecommandlockouts

\newcommand{\lb}{\left(}
\newcommand{\rb}{\right)}

\DeclareFontFamily{U}{matha}{\hyphenchar\font45}
\DeclareFontShape{U}{matha}{m}{n}{
	<5> <6> <7> <8> <9> <10> gen * matha
	<10.95> matha10 <12> <14.4> <17.28> <20.74> <24.88> matha12
}{}
\DeclareSymbolFont{matha}{U}{matha}{m}{n}
\DeclareFontSubstitution{U}{matha}{m}{n}

\DeclareFontFamily{U}{mathx}{\hyphenchar\font45}
\DeclareFontShape{U}{mathx}{m}{n}{
	<5> <6> <7> <8> <9> <10>
	<10.95> <12> <14.4> <17.28> <20.74> <24.88>
	mathx10
}{}
\DeclareSymbolFont{mathx}{U}{mathx}{m}{n}
\DeclareFontSubstitution{U}{mathx}{m}{n}

\DeclareMathDelimiter{\vvvert}{0}{matha}{"7E}{mathx}{"17}

\newtheorem{defn}{{Definition}}
\newtheorem{theorem}{{Theorem}}
\newtheorem{lemma}{{Lemma}}

\newtheorem{proposition}{Proposition}
\newtheorem{corollary}{{\bf Corollary}}

\newcommand{\diag}{\mathrm{diag}}

\newcommand{\D}{\mathcal{D}}

\newcommand{\LL}{\mathcal{L}}
\newcommand{\DD}{\mathcal{D}}
\newcommand{\PP}{\mathbb{P}}
\newcommand{\EE}{\mathbb{E}}
\newcommand{\setS}{\mathcal{S}}
\newcommand{\setA}{\mathcal{A}}
\newcommand{\setB}{\mathcal{B}}
\newcommand{\setG}{\mathcal{G}}
\newcommand{\setP}{\mathcal{P}}
\newcommand{\Real}{\mathbb{R}}

\newcommand{\ls}{\left[}
\newcommand{\rs}{\right]}
\newcommand{\lc}{\left\{}
\newcommand{\rc}{\right\}}
\newcommand{\lv}{\left|}
\newcommand{\rv}{\right|}
\newcommand{\lvv}{\|}
\newcommand{\rvv}{\|}
\newcommand{\lvvv}{\left \vvvert}
\newcommand{\rvvv}{\right \vvvert}
\newcommand{\veca}{\mathbf{a}}

\newcommand{\vecx}{\mathbf{x}}
\newcommand{\vecz}{\mathbf{z}}
\newcommand{\vecy}{\mathbf{y}}
\newcommand{\vecv}{\mathbf{v}}
\newcommand{\vecu}{\mathbf{u}}
\newcommand{\vecb}{\mathbf{b}}
\newcommand{\vecw}{\mathbf{w}}
\newcommand{\matA}{\mathbf{A}}
\newcommand{\matB}{\mathbf{B}}
\newcommand{\matI}{\mathbf{I}}
\newcommand{\matC}{\mathbf{C}}

\newcommand{\matU}{\mathbf{U}}
\newcommand{\matW}{\mathbf{W}}
\newcommand{\matR}{\mathbf{R}}
\newcommand{\matH}{\mathbf{H}}
\newcommand{\matX}{\mathbf{X}}
\newcommand{\matY}{\mathbf{Y}}
\newcommand{\matZ}{\mathbf{Z}}

\newcommand{\matSigma}{\mathbf{\Sigma}}
\newcommand{\matGamma}{\mathbf{\Gamma}}
\newcommand{\matLambda}{\mathbf{\Lambda}}
\newcommand{\vgamma}{\boldsymbol{\gamma}}

\newcommand{\vecyall}{\vecy_{1}, \vecy_{2}, \dots, \vecy_{L}}
\newcommand{\vecxall}{\vecx_{1}, \vecx_{2}, \dots, \vecx_{L}}
\newcommand{\covy}{\sigma^{2}\mathbf{I}_{m} + \mathbf{A} \mathbf{\Gamma} \mathbf{A}^{T}} 
\newcommand{\vgammamax}{\boldsymbol{\gamma}_{\text{max}}}
\newcommand{\vgammamin}{\boldsymbol{\gamma}_{\text{min}}}
\allowdisplaybreaks

\begin{document}

\title{On the Support Recovery of Jointly Sparse Gaussian Sources Via Sparse Bayesian Learning}

\author{
\authorblockN{Saurabh Khanna, \IEEEmembership{Member, IEEE} and Chandra R. Murthy, \IEEEmembership{Senior Member, IEEE}}\\
\thanks{S. Khanna is currently with the Department of Electronics and Communication Engineering, Indian Institute of Technology Roorkee, India (e-mail: sakhanna@ece.iitr.ac.in).}
\thanks{C. R. Murthy is with the Department of Electrical Communication Engineering (ECE), Indian Institute of Science (IISc), Bangalore, India (e-mail: cmurthy@iisc.ac.in).}
\thanks{The main part of this work was carried out when S. Khanna was at ECE Department, IISc.}

}


\maketitle

\begin{abstract}
In this work, we provide non-asymptotic, probabilistic guarantees for successful recovery of the common nonzero support of jointly sparse Gaussian sources in the multiple measurement vector (MMV) problem. 
The support recovery problem is formulated as the marginalized maximum likelihood (or type-II ML) estimation of the variance hyperparameters of a joint sparsity inducing Gaussian prior on the source signals. 
We derive conditions under which the resulting nonconvex constrained optimization perfectly recovers the 
nonzero support of a joint-sparse Gaussian source ensemble with arbitrarily high probability. 
The support error probability decays exponentially with the number of MMVs at a 
rate that depends on the smallest restricted singular value and the nonnegative null space property of the self 
Khatri-Rao product of the sensing matrix. Our analysis confirms that nonzero supports of size as high as 
$O(m^2)$ are recoverable from $m$ measurements per sparse vector. 
\textcolor{black}{Our derived sufficient conditions for support consistency of the proposed} constrained type-II ML solution \textcolor{black}{also guarantee the support consistency} of any global solution of the multiple sparse Bayesian learning (M-SBL) optimization whose nonzero coefficients lie inside a bounded interval. 
For the case of noiseless measurements, we \textcolor{black}{further} show that a single 
MMV is sufficient for perfect recovery of the $k$-sparse support by M-SBL, provided all subsets of $k + 1$ columns of the sensing matrix are linearly independent. 
\end{abstract}
\begin{keywords}
Compressive sensing, support recovery, sparse Bayesian Learning, 
joint sparsity, Restricted Isometry Property, Khatri-Rao product.
\end{keywords}

\section{Introduction}
Joint sparsity has emerged as an important and versatile signal structure in 
the field of sparse signal processing.
Two or more vectors are said to be jointly sparse if their nonzero 
coefficients belong to the same index set, i.e., the vectors share a common nonzero support. 
Joint sparsity arises naturally in multi-modal or multi-channel analysis of signals residing in low dimensional signal subspaces. 
The underlying joint sparsity of signals can prove useful in resolving ambiguities in the common support that may arise due to erroneous estimation of the support of the individual sparse signal vectors from noisy measurements. 
This viewpoint has been successfully exploited in several real-world applications such as MIMO channel estimation 
\cite{Masood15SparseMIMOChEst, Barbotin12MIMOJSM, Ranjitha15SBLChEst}, 
distributed source coding \cite{Duarte09DCS_main, HaoZhu09CAMoM}, 
multi-task compressive sensing \cite{Shihao09MultiTaskCS}, 
distributed event localization \cite{Erseghe15MLthruADMM}, array 
signal processing \cite{QiaoPal19MMV}, cooperative spectrum sensing~\cite{Tian11DistSpectrumSensing, Geert09CompressiveSpectrumSensing, Bazerque_10_dss}, and 
user activity detection in massive machine-type communications \cite{alex2019nonbayesian}.

In the sparse signal recovery literature, the estimation of \textcolor{black}{jointly} sparse signals is commonly referred to as the \emph{multiple measurement vector} (MMV) 
problem \cite{Cotter_05_mmv} where the signal of interest is a matrix $\matX \in \Real^{n \times L}$ 
whose columns are \textcolor{black}{jointly} sparse vectors in $\Real^{n}$. 
As a result, $\matX$ is a 
row sparse matrix with only a fraction of its rows containing 
nonzero elements and the rest of the rows made up entirely of 
zeros. In the MMV problem, the goal is to recover $\matX$ from 
its noisy, linear measurements $\matY \in \Real^{m \times L}$. 
The measurement matrix $\matY$ (each column is called a single measurement vector (SMV)) is generated as 
\begin{equation} \label{mmv_meas_model_0}
\matY = \matA \matX + \matW, 
\end{equation}
where $\matA \in \Real^{m \times n}$ 
is a known sensing matrix and $\matW \in \Real^{m \times L}$ models the unknown noise in the measurements.
For $m < n$, the above linear system is underdetermined and has infinitely many solutions for $\matX$. However, if $\matA$ satisfies certain restricted isometry properties, a unique row-sparse solution can {still} be guaranteed \cite{ChenHuo06MMV, Cotter_05_mmv, Tropp_05_SOMP, FoucartBook13}.
\ifdefined \SKIPTEMP
When $L = 1$, the MMV problem reverts to the SMV problem, for which a unique $k$-sparse 
solution under the minimum $\ell_{0}$ norm criterion is guaranteed if and only 
if $k < \text{spark}(\matA)/2$, where 
$\text{spark}(\matA)$ is the minimum number of linearly 
dependent columns in $\matA$~\cite{FoucartBook13}. This result is generalized to the noiseless MMV case 
in \cite{ChenHuo06MMV}, where it is shown that any $\matX$ with 
$k$ or fewer nonzero rows and satisfying \eqref{mmv_meas_model_0} is a unique solution 
under the minimum $\ell_{0}$ norm criterion,
if $k < \lceil \text{spark}(\matA) -1 + \text{rank}(\matY) \rceil /2$. 
This suggests that a trade off exists between the number of MMVs, $L$, and the number of measurements, $m$, 
which can be optimized to reduce the overall measurement acquisition cost in an 
MMV problem.
\fi
In many applications, one is primarily interested in identifying the nonzero rows of $\matX$. 
This gives rise to the \textit{joint-sparse support recovery (JSSR)} problem where the goal is to recover the row support of $\matX$ given $\matY$ and~$\matA$. 
Interestingly, unlike the nonzero coefficients in a $k$-row-sparse~$\matX$ which can be robustly recovered only if $m \ge k$, 
its nonzero row support can be recovered correctly even from $m < k$ 
measurements. 

The vast majority of the JSSR solvers \cite{RaoCotterDiversityMin04, Cotter_05_mmv, Tropp_05_SOMP, Blanchard14Greedy, Bresler12SAMUSIC, JohnChulYe12CSMUSIC, Mishali08Rembo} 
implicitly assume that the number of nonzero rows in~$\matX$ is less than the number of measurements per SMV, i.e., $k < m$. Recently, \textit{correlation-aware} solvers such as Co-LASSO 
\cite{PPal15PushingTheLimits} and M-SBL \cite{Wipf_07_msbl, 
Balkan14JSBLResults}, and their extensions, e.g., RD-CMP~\cite{Khanna17RDCMP}, 
Co-LASSO-EXPGRD \cite{Sak18EXPGRD} and MRNNQP \cite{Lekshmi18CovEst}
have been empirically shown to be capable of recovering sparse supports of size $k > m$ when $L$ is large. 
However, a complete theoretical understanding of when these methods are guaranteed to be successful 
is still lacking. Especially for M-SBL, due to its nonseparable and nonconvex 
log-marginalized likelihood objective, any support consistency 
guarantees about its global or local maxima are currently known only under restrictive assumptions, e.g., for noiseless measurements and row orthogonality of the signal matrix~$\matX$ 
in \cite{Wipf_07_msbl, Balkan14JSBLResults}, for the SMV case and row 
orthogonality of the measurement matrix in~\cite{YeeAtchade17SBLGuarantees}, 
and under the assumption of prior knowledge of the true support size and source 
signal power in \cite{NehoraiGTang10, Koochakzadeh18CorrAwareMMV, alex2019nonbayesian}.

The main goal of the present work is to identify sufficient conditions for exact support recovery 
\textcolor{black}{ in the JSSR problem by a \textit{correlation-aware} procedure based on the M-SBL optimization.}   
This is accomplished by analyzing the support consistency of \textcolor{black}{gobal maxima of a constrained variant of M-SBL's type-II maximum-likelihood optimization.} 
Our newly derived sufficient conditions \textcolor{black}{for exact recovery of $k$-sized supports in the JSSR 
problem} confirm that the support error probability vanishes 
even for $k > m$ when the number of MMVs is large enough. 

\subsection{Existing results on support recovery using MMVs}

Early theoretical works focussed on obtaining guarantees for a unique joint-sparse solution to the canonical $\ell_{0}$ norm minimization problem:
\begin{equation}
	L_{0}: \;\;\;\;
	\underset{\matX \in \Real^{n \times L}}{\text{min }} \lvv \matX \rvv_{0} 
	\;\;\;\; \text{s.t. } \matA \matX = \matY,
	\label{lo_min_prob}
\end{equation}
where $\lvv \matX \rvv_{0}$ denotes the number of nonzero rows in~$\matX$. In~\cite{Cotter_05_mmv, ChenHuo06MMV}, the authors showed that 
the $L_{0}$ problem admits a unique $k$-sparse solution provided
$k < \lceil \lb \text{spark}(\matA) - 1 + \text{rank}(\matY) \rb / 2 \rceil$, 
where spark($\matA$) denotes the smallest integer $p$ such that there exist $p$ linearly dependent columns in $\matA$. 
This result establishes that the SMV bottleneck of $k <  m/2 $ for $\ell_{0}$ norm 
based support recovery can be overcome by using multiple measurement vectors. 
Furthermore, the sparsity bound suggests that supports of size $k < m$ 
are uniquely recoverable. 

To circumvent the combinatorial hardness of the $L_{0}$ problem, 
\cite{RaoCotterDiversityMin04} proposes to minimize the 
$\ell_{p,q}$ mixed-norm of~$\matX$ instead of the $\ell_{0}$ norm. 
The $\ell_{p,q}$ norm of $\matX$ is evaluated as $\lvv \matX \rvv_{p,q} \triangleq 
\lb \sum_{i = 1}^{m} \lvv \matX(i, :) \rvv_{p}^{q} \rb^{1/q} $.
Variants of the $\ell_{p,q}$ norm minimization problem with different combinations 
of $p$ and $q$ have been investigated independently in \cite{Cotter_05_mmv, ChenHuo06MMV, BergFriedlander10, Eldar09UnionOfSubspaces}. 
For $p \ge 1, q = 1$, \cite{ChenHuo06MMV} has shown that $\ell_{p,q}$ norm minimization problem 
has a unique $k$-sparse solution, provided $\matA$ satisfies
$\lvv \matA_{\setS}^{\dagger} \veca_{j}\rvv_{1} < 1$, for all 
$j \notin \setS$ and for all $\setS \subset [n], |\setS| \le k$, 
where $\matA_{\setS}^{\dagger} = 
\lb \matA_{\setS}^{T} \matA_{\setS} \rb^{-1} \matA_{\setS}^{T}$.
\ifdefined \VERBOSE
In fact, this condition is also the sufficient condition for 
uniqueness under minimum 
$\lvv \lc f\lb \matX(i, :) \rb \rc_{i = 1}^{n} \rvv_{1}$ criterion, 
where $f(\cdot)$ is any arbitrary inner vector norm. 
\fi
\ifdefined \SKIPTEMP
In \cite{BergFriedlander10}, a necessary condition for a unique joint-sparse solution 
to the $L_{p, q}$ problem is identified. 
It is shown that any $\matX_{0}$ with 
$\text{supp}(\matX_{0}) \subseteq \setS$ and satisfying $\matY = \matA \matX_{0}$
is a unique solution to~\eqref{lpq_min_prob} 
if and only if any $\matZ \neq \mathbf{0}$ with columns belonging to the null space of $\matA$ 
satisfies $\sum_{i \in \setS} \lvv \matZ(i, :) \rvv < 
\sum_{i' \notin \setS} \lvv \matZ(i', :) \rvv$ for any vector norm 
$\lvv \; \cdot \; \rvv$. 
However, it is combinatorially hard to verify such a condition. 
\fi
This also serves as a sufficient condition for exact support recovery via
simultaneous orthogonal matching pursuit (SOMP) \cite{Tropp_05_SOMP}, a greedy support recovery algorithm. 
In \cite{Blanchard14Greedy}, support recovery performance of 
various correlation based greedy and iterative hard-thresholding type algorithms 
is studied in the noiseless MMV setup. The sufficient conditions for exact support 
recovery are specified in terms of the 
asymmetric restricted isometry constants of the sensing matrix.

A common limitation of the aforementioned support recovery techniques is that 
they are capable of uniquely recovering supports of size up to only 
$k < m/2$.
In \cite{DaviesEldar12RankAware}, rank aware OMP and rank aware order recursive 
matching pursuit are shown to perfectly recover any $k$-sized support from noiseless measurements 
as long as $k < \text{spark}(\matA) -1$ and $\text{rank}(\matX) = k$. 
For the rank defective case, i.e.,  $\text{rank}(\matX) < k$, compressed sensing MUSIC~\cite{JohnChulYe12CSMUSIC} and 
subspace-augmented MUSIC~\cite{Bresler12SAMUSIC} are still capable of recovering any $k < \text{spark}(\matA)-1$ sized support 
as long as partial support of size $k - \text{rank}(\matX)$ can be estimated by some 
other support recovery algorithm. 

In \cite{NehoraiGTang10}, the MMV support recovery problem is formulated as a multiple hypothesis 
testing problem. 
Necessary and sufficient conditions for perfect support recovery with high 
probability are derived under the assumption that the columns of $\matX$ are 
i.i.d. $\mathcal{N}(0, \diag{(\mathbf{1}_{\setS^*})})$, where $\setS^*$ denotes 
the true support set of a known size $k$. An 
exponential decay rate of the support error probability is derived in closed 
form, which is however not conducive to general interpretation or deeper 
analysis. For the particular case of randomly constructed Gaussian sensing 
matrix $\matA$ with $m = \Omega\lb k \log \frac{n}{k}\rb$ rows, it is shown 
that $L \gg \frac{\log n}{\log \log n}$ suffices for diminishing support error 
probability with increasing $L$. One of our contributions in 
the present work is to extend this result to a more general signal prior 
on~$\matX$ and show that the support error probability vanishes even if $m$ 
scales sublinearly in the support size $k$.

In \cite{YuzheJin13MMV}, the support recovery problem is analyzed as a single-input-multi-output 
MAC communication problem. For number of nonzero rows fixed to $k$, $m = 
\Omega(\frac{k \log n}{c(\matX)})$ is shown to 
be both necessary and sufficient for successful support recovery as the problem size tends to 
infinity. The quantity $c(\matX)$ is a capacity like term \textcolor{black}{that} depends on the elements of the nonzero rows in $\matX$ 
and the noise power. Even fewer measurements $m = \Omega \lb \frac{k}{L}\log{n}\rb$ suffices when each measurement vector is generated using an independently drawn sensing matrix~\cite{ParkLee17DiffMatInMMV}. 
The $m < k$ regime has been studied in \cite{Lekshmi19MMVboundsForJSSR}, and it was shown that $L = \Theta\lb \frac{k^2}{m^2} \log(k(n-k)) \rb$ is both necessary and sufficient to exactly recover the row-support of $\matX$.

\textcolor{black}{
Aside from the problem of inferring the common support of joint sparse vectors from MMVs, \cite{ZhuBaron18, ZhuBaron13NoisyMMVMSE, ZhuBaron17MMVMSE} use replica analysis to obtain a variational characterization of the minimum mean squared error incurred while reconstructing the whole joint-sparse vectors from the noisy MMVs, as the dimension $n$ scales to infinity while keeping the measurement-rate $m/n$ finite.    
}


\ifdefined \SKIPTEMP
A major limitation the above algorithms 
is that their worst case performance\footnote{For example, 
when all columns of the signal matrix $\matX$ are the same.} does not improve with $\text{rank}(\matY)$. 
Therefore, these algorithms are rightfully called \textit{rank-blind} \cite{DaviesEldar12RankAware}, 
as their exact support recovery conditions are independent of $L$, the 
number of MMVs. However, their average case performances do improve 
with an increase in the number of MMVs \cite{Gribonval08AtomsUnite}. 
These rank-blind methods also fail to meet the $\ell_{0}$-norm uniqueness criterion, i.e., 
$k < m + 1$, and can only guarantee a unique $k$-sparse solution if $k < m/2$. 
\fi

\ifdefined \SKIPTEMP
If the nonzero rows of $\matX$ are linearly dependent, then $k \textcolor{black}{\;>\;} \text{rank}(\matX)$. 
In this underdetermined case, exact support recovery is still possible 
if partial support set is known \cite{JohnChulYe12CSMUSIC, Bresler12SAMUSIC}.
In \cite{JohnChulYe12CSMUSIC}, the authors proposed the CS-MUSIC algorithm 
which can recover any $k \ge \text{rank}(\matX)$ sized support 
as long as $k - \text{rank}(\matX)$ partial support can be estimated by another sparse signal recovery algorithm. 
Subspace-Augmented MUSIC (SA-MUSIC) \cite{Bresler12SAMUSIC} is another algorithm based on a similar idea. 
\fi

\ifdefined \VERBOSE
Finally, there are Bayesian inference MMV algorithms which output a posterior distribution of 
the unknown $\matX$, instead of a direct estimate. In a Bayesian approach, $\matX$ is considered 
to be a random variable drawn from an unknown but parameterized joint sparsity inducing signal prior. 
The prior which has maximal Bayesian evidence with respect to the observations $\matY$ \cite{MacKay91bayesianinterpolation} 
is finally used to compute a maximum a posteriori probability (MAP) estimate of $\matX$.
The signal prior acts as a proxy for the regularization on the solution space akin to the penalty based methods.
Two popular signal priors used in the MMV literature are the Bernoulli-Gaussian prior (also known 
as the spike and slab prior) \cite{Ziniel13AMPMMV} and the Gaussian prior \cite{Wipf_07_msbl}. 
Learning the signal prior/model which best explains the observations is the central idea in the Bayesian approach, 
and algorithms inspired by this approach are said to possess the \textit{automatic relevance determination} property \cite{MacKay91bayesianinterpolation}.
\fi

\ifdefined \SKIPTEMP
In the above discussion, in particular about algorithm-specific support recovery guarantees, 
the number of nonzero rows in $\matX$ is typically assumed to be less than $m$, the number of measurements per MMV. We now review the \textcolor{black}{recently proposed} covariance matching approach, which is capable of recovering supports of size $k$ greater than~$m$.
\fi

\subsection{Correlation-aware support recovery}
A key insight was propounded in \cite{PPal15PushingTheLimits}, 
that there often exists a latent structure in the MMV problem: \textit{the nonzero rows of~$\matX$ are uncorrelated}. This signal structure can be enforced by modeling each column of 
$\matX$ to be i.i.d. $\mathcal{N}(0, \diag{(\vgamma)})$, where $\vgamma \in \Real^{n}_{+}$ 
is a nonnegative vector of variance parameters. 
Under this source model, identifying the nonzero rows of~$\matX$ amounts 
to detecting the support of $\vgamma$. In \cite{PPal15PushingTheLimits}, Co-LASSO is proposed to recover $\vgamma$.
Instead of directly working with the linear observations~$\matY$, 
Co-LASSO uses their covariance form, $\frac{1}{L}\matY \matY^{T}$, 
as input, and estimates $\vgamma$ as a solution of the following constrained 
$\ell_{1}$-norm minimization problem:
\begin{equation}
\underset{\vgamma \in \Real^{n}_{+}}{\text{min }} \lvv \vgamma \rvv_{1} \;\;
\text{s.t. } 
(\matA \odot \matA)\vgamma = \text{vec} \lb \frac{1}{L}\matY \matY^{T} \rb, 
\label{defn_cov_mmv}
\end{equation}
where $\matA \odot \matA$ denotes the Khatri-Rao product (i.e., the columnwise Kronecker 
product) of $\matA$ with itself. The linear constraints in 
\eqref{defn_cov_mmv} depict the second-order moment-matching constraints, 
specifically, the covariance matching equation: 
$\frac{1}{L}\matY \matY^{T} \approx \matA \diag{\lb \vgamma \rb}\matA^{T}$. 
Since \eqref{defn_cov_mmv} comprises up to $(m^{2} + m)/2$ linearly independent 
equations in $\vgamma$, sparsity levels as 
high as $O(m^2)$ are potentially recoverable. To recover the maximum level of 
sparsity, $k = (m^2 + m)/2$, 
a necessary condition derived in \cite{PPal15PushingTheLimits, PPal14SBLParamIdentifiability} 
dictates that the columnwise self Khatri-Rao product matrix 
$\matA \odot \matA$ must have full Kruskal rank,\footnote{
The Kruskal rank of an $m \times n$ matrix $\matA$ is the largest integer $k$ such that any $k$ columns of $\matA$ are linearly independent.
} i.e., $\text{Krank}(\matA \odot \matA) = (m^2 + m)/2$. 
Another popular MMV algorithm, M-SBL \cite{Wipf_07_msbl},  
also imposes a common Gaussian prior $\mathcal{N}(0, \diag{(\vgamma)})$ on the 
columns of $\matX$ and hence implicitly exploits the latent uncorrelatedness of 
the nonzero entries in~$\matX$.
Interestingly, similar to Co-LASSO, \textcolor{black}{the performance of Sparse Bayesian Learning (SBL)-based support recovery methods depend on the null-space structure of the self Khatri-Rao product $\matA 
\odot \matA$.} Making this connection explicit is one of the 
goals of this work. 

The sparsistency of \textcolor{black}{different types of constrained solutions of the M-SBL optimization} has been earlier investigated in \cite{Wipf_07_msbl} for $k < m$, and in \cite{Balkan14JSBLResults, alex2019nonbayesian} for $k \ge m$. In \cite{alex2019nonbayesian}, it is assumed that support size and the source signal powers are known a priori, while both \cite{Wipf_07_msbl} and \cite{Balkan14JSBLResults} assume that the measurements are noiseless, and the nonzero rows of $\matX$ are orthogonal. For finite $L$, the 
row-orthogonality condition is too restrictive for a deterministic $\matX$ and 
almost never true for a random $\matX$ drawn from a continuous distribution. 
In \cite{WeiYu19CovMatchPHTR}, necessary and sufficient conditions for amplitude-wise consistency of the M-SBL solution have been identified for $L \to \infty$.
In contrast, we study the non-asymptotic case while taking the measurement noise into consideration, and also dispense with the restrictive row orthogonality condition on $\matX$. 
Moreover, unlike \cite{alex2019nonbayesian}, our analysis does not require the source signal powers to be known upfront.

\subsection{Our contributions}

\ifdefined \SKIPTEMP
It is shown in \cite[Theorem 1]{Wipf_07_msbl} that for $k < \text{spark}(\matA) -1$, and for noiseless 
measurements $\matY$ generated from any $k$-sparse $\matX$ with orthogonal nonzero rows, M-SBL exactly 
recovers the true row support of $\matX$. 
Also, for $m \le k \le n$, in the noiseless setting, M-SBL exactly recovers any $k$-sparse support, if~\cite{Balkan14JSBLResults}
\begin{enumerate}[label=(\roman*)]
\item the $k$ nonzero rows of $\matX$ are mutually orthogonal, and 
\item $\text{rank}(\matA \odot \matA) = n$.
\end{enumerate}

In practice, the orthogonality condition for the nonzero rows of $\matX$ is almost never met unless there are infinitely many MMVs. 
Hence, the existing support recovery conditions for M-SBL are significant mainly from a theoretical perspective. 
In this paper, we present new non-asymptotic, practically useful performance guarantees for the M-SBL algorithm.  
\fi

\begin{enumerate}
\item We interpret the M-SBL optimization as a Bregman matrix divergence minimization problem; opening up new avenues to exploit the vast literature on Bregman divergence minimization towards devising faster, more robust 
algorithms for support recovery.
\item \textcolor{black}{For the JSSR problem, we analyze the consistency of the estimated nonzero support inferred} from the variance hyperparameters of a zero mean, row-sparsity inducing Gaussian prior on~$\matX$. These hyperparameters are 
estimated via a constrained type-II maximum likelihood (ML) procedure, \textcolor{black}{called cM-SBL, wherein}  
the nonzero coefficients of the ML estimate are constrained to lie inside a bounded 
interval. We show that the support error probability decays exponentially with number of MMVs, and the error exponent is related to the null space 
and restricted singular value properties of $\matA \odot \matA$, the self Khatri-Rao product of 
the sensing matrix $\matA$ with itself. Explicit bounds on the number of 
MMVs sufficient for vanishing support error probability for both noisy 
and nearly noiseless measurements are derived (in Theorem~\ref{thm_suff_conditions_form1}). 
The support \textcolor{black}{consistency} guarantees obtained for the constrained type-II ML \textcolor{black}{solution} also apply to any \textcolor{black}{global} solution of the M-SBL optimization whose nonzero coefficients lie inside a known interval (see Corollary~\ref{corr_transfer_results_to_msbl}).  

\item In the special case where the sensing matrix $\matA$ is constructed from i.i.d. $\mathcal{N}\lb 0, \frac{1}{m} \rb$ entries, we bound the MMV complexity for which the \textcolor{black}{cM-SBL} optimization exactly recovers the true $k$-sparse support using $m = \Omega(\sqrt{k} \log n)$ measurements per MMV. In the case of noiseless measurements, we show that M-SBL exactly recovers the true $k$-sparse support from a single measurement vector, provided $k < \text{spark}(\matA) -1$.

\end{enumerate}
\ifdefined \VERBOSE
Compared to the previous results \cite{Wipf_07_msbl, Balkan14JSBLResults}, our sufficient conditions dispenses with 
the restrictive orthogonality assumption for the nonzero rows of $\matX$. This is replaced with a milder assumption 
on the signal matrix $\matX$ being Gaussian distributed. 
\fi
\ifdefined \VERBOSE
Compared to an earlier work in \cite{NehoraiGTang10} which also considers a Gaussian signal prior on $\matX$, 
but with a binary valued covariance matrix, we consider a more general setting where the variance parameters are continuous valued.
\fi
A key aspect of our results is that our sufficient conditions are expressed in terms of number of MMVs and properties of the sensing matrix $\matA$. This makes our results applicable to both random as well as deterministic constructions of $\matA$. As part of our analysis, we present a new lower bound for the $\frac{1}{2}$-R\'enyi divergence between a pair of multivariate Gaussian densities (in Proposition~\ref{prop_renyi_div_lb_analytical}), and an interesting null space property of Khatri-Rao product matrices (in Theorem~\ref{thm_strong_rnsp_kr}), which may be of independent interest.


The remainder of the paper is organized as follows. In section \ref{sec:notion_model_prelim}, we formulate the JSSR problem and introduce our source model for $\matX$. We also review the M-SBL 
algorithm~\cite{Wipf_07_msbl} and interpret the M-SBL cost function as a 
Bregman matrix divergence. In section \ref{sec:prelim_concepts}, we cover some 
preliminary concepts \textcolor{black}{that} are used while analyzing the support recovery performance of the constrained and unconstrained variant of the SBL procedure. In section \ref{sec:abstract_perr_bound_derivation}, we derive an abstract upper bound for the support 
error probability, which is used in section \ref{sec:sufficient_conditions} to derive our main result, namely, the sufficient conditions for vanishing support error 
probability in M-SBL. 
In section \ref{sec:discussion}, we discuss the implications of the new results in the context of several interesting special cases. Our final conclusions are presented in section~\ref{sec:conclusion}.

\subsection{Notation} \label{sec:notations}
Throughout this paper, scalar variables are denoted by lowercase alphabets and vectors are denoted by boldface lowercase alphabets. Matrices are denoted by boldface uppercase alphabets 
and calligraphic uppercase alphabets denote sets. 

Given a vector $\vecx$, $\vecx(i)$ represents its $i^{\text{th}}$ entry. 
$\text{supp}(\vecx)$ denotes the support of $\vecx$, the set of indices 
corresponding to nonzero entries in $\vecx$. Likewise, 
$\mathcal{R}(\matX)$ denotes the set of indices of all nonzero rows in $\matX$ and 
is called the row-support of $\matX$.
\textcolor{black}{For any $n \in \mathbb{N}$, $[n]\triangleq \{1, 2, \ldots, N\}$.}
For any $n$ dimensional vector $\vecx$ and index set $\setS \subseteq [n]$, the vector $\vecx_{\setS}$ 
is an $|\setS| \times 1$ sized vector retaining only those entries of $\vecx$ \textcolor{black}{that} are indexed by elements of $\setS$. 
Likewise, $\matA_{\setS}$ is a submatrix comprising the columns of $\matA$ indexed by $\setS$. 
$\text{Null}(\matA)$ and $\text{Col}(\matA)$ denote the 
null space and column space of the matrix $\matA$, respectively. 
The spectral, Frobenius and maximum absolute row sum matrix norms of $\matA$ 
are denoted by $\lvvv \matA \rvvv_{2}$, $\lvv \matA \rvv_{F}$, and $\lvvv \matA \rvvv_{\infty}$, respectively. 
$\PP(\mathcal{E})$ denotes the probability of event $\mathcal{E}$. $\mathcal{N}(\mu, \matSigma)$ denotes 
the Gaussian probability density with mean $\mu$ and covariance matrix $\matSigma$. 
For any square matrix $\matC$, $\text{tr}(\matC)$ and $\lv\matC \rv$ denote its 
trace and determinant, respectively. $S_{++}^{n}$ denotes the set of all $n \times n$ positive definite matrices. 

Given positive sequences $\lc f_{n} \rc_{n = 1}^{\infty}$ and $\lc g_{n} \rc_{n = 1}^{\infty}$, $g_{n} = O(f_{n})$ denotes that there exists an $N \in \mathbb{N}$ and a universal constant $C > 0$ such that $g_{n} \le C f_{n}$ for $n \ge N$. Similarly, $g_{n} = \Omega(f_{n})$ denotes that $g_{n} \ge c f_{n}$ for $n \ge N$ and some  constant $c > 0$. Lastly, $f_{n} = \Theta(g_{n})$ implies that $c f_{n} \le g_{n} \le C f_{n}$ for $n \ge N$.

\section{System Model and the M-SBL Algorithm} \label{sec:notion_model_prelim}

\subsection{Joint-sparse support recovery (JSSR)} \label{sec:joint_sparse_supp_recov_prob}
We now formally state the JSSR problem. Suppose $\vecxall$ are $L$ distinct 
joint-sparse vectors in $\Real^{n}$ 
with a common nonzero support denoted by the index set $\setS^{*} \subseteq [n]$. 
Let $K$ be the maximum size of the common support, i.e., $\lv \setS^{*} \rv \le K$. 
In JSSR, we are interested in recovering $\setS^{*}$ from  noisy underdetermined linear measurements $\vecyall$ generated as 
\begin{equation} \label{mmv_meas_model}
\vecy_{j} = \matA \vecx_{j} + \vecw_{j}, \;\;\; 1 \le j \le L.
\end{equation}
We assume that the sensing matrix $\matA \in \Real^{m \times n}$ is a non-degenerate matrix with $m \le n$, i.e., 
any $m$ columns of $\matA$ are linearly independent, or $\text{spark}(\matA) = m + 1$. 
The noise vector $\vecw \in \Real^{m}$ is zero mean Gaussian distributed with diagonal covariance matrix $\sigma^{2} \matI_{m}$. 
The linear measurement model in \eqref{mmv_meas_model} can be rewritten in a compact MMV form as 
$\matY = \matA \matX + \matW$,
where $\matY = \ls \vecyall \rs$, $\matX = \ls \vecxall \rs$ and 
$\matW = \ls \vecw_{1}, \vecw_{2}, \dots, \vecw_{L} \rs$ are 
the observation, signal and noise matrices, respectively. 
Since the columns of $\matX$ are jointly sparse with common support $\setS^{*}$,  
$\matX$ is a row sparse matrix with row support $\mathcal{R}(\matX) = \setS^{*}$. 

\subsection{Gaussian source assumption} \label{sec:signal_model}
We assume that if the $i^{\text{th}}$ row of the unknown signal 
matrix~$\matX$ is nonzero, then it is a Gaussian ensemble of $L$ i.i.d.\ zero mean random variables with 
a common variance $\vgamma^*(i)$ which lies in the interval $[\vgammamin, \vgammamax]$. We refer to this as \textbf{Assumption (A1)}. 
%
%
An immediate consequence of (\textbf{A1}) is that there exists a 
bounded, nonnegative, and at most $K$ sparse vector, 
$\vgamma^{*} \in \Real^{n}_{+}$, such that the columns $\vecx_{j}$ are i.i.d.\ $\mathcal{N}(0, \matGamma^{*})$ with 
$\matGamma^{*} \triangleq \diag{(\vgamma^{*})}$. 
Furthermore, $\mathcal{R}(\matX)$ and $\text{supp}(\vgamma^{*})$ are the same 
and both are equal to $\setS^{*}$. 
\ifdefined \SKIPTEMP
and hence $\matX$ adheres to the joint sparsity model JSM-2 in~\cite{Duarte09DCS_main}. 
\fi

\ifdefined \SKIPTEMP
Compared to \cite{NehoraiGTang10} which also considers a Gaussian source model but with binary valued variance parameters, 
ours is a more general source model which allows for real valued and distinct row-variances. 
In fact, the signal model in \cite{NehoraiGTang10} can be obtained as a special case by setting $\vgammamin = \vgammamax = 1$ 
in the assumption (A1). 
\fi


\subsection{Multiple sparse Bayesian learning (M-SBL)} \label{sec:msbl_overview} 
We now review M-SBL \cite{Wipf_07_msbl}, a type-II maximum likelihood (ML) 
procedure for estimating joint-sparse signals from compressive linear measurements. 
M-SBL models the columns of $\matX$ as i.i.d.\ $\mathcal{N}(0, \matGamma)$, 
%
where $\matGamma = \diag{(\vgamma)}$, and $\vgamma = \ls \vgamma(1), \vgamma(2), \dots, \vgamma(n) \rs^{T}$ 
is an $n \times 1$ vector of unknown nonnegative variance parameters. The elements of $\vgamma$ are collectively called hyperparameters as 
they represent the parameters of the signal prior. 
Since the hyperparameter~$\vgamma(i)$ models the common variance of the $i$th 
row of $\matX$, 
if $\vgamma(i) = 0$, it drives the posterior estimate of $\vecx_{j}(i)$ to zero for $1 \le j \le L$. 
Consequently, if $\vgamma$ is estimated to be a sparse vector, it induces joint 
sparsity in $\matX$.

In M-SBL, the hyperparameter vector $\vgamma$ is chosen to maximize the 
Bayesian evidence $p(\matY ; \vgamma)$, which is tantamount to finding 
the ML estimate of $\vgamma$. Let $\hat{\vgamma}_{\text{ML}}$ denote the ML estimate of $\vgamma$, i.e.,
\begin{equation} \label{ml_estimate_of_vgamma}
	\hat{\vgamma} = \underset{\vgamma \in \Real^{n}_{+}}{\text{arg max }} \log{p(\matY ; \vgamma)}.
\end{equation}
The Gaussian prior on $\vecx_{j}$ combined with the linear measurement model induces Gaussian observations, i.e., 
$p(\vecy_{j}; \vgamma) = \mathcal{N}(0, \sigma^{2} \matI_{m} + \matA \matGamma \matA^{T})$. 
For a fixed $\vgamma$, the MMVs $\vecy_{j}$ are mutually independent. Hence, it follows that 
\begin{eqnarray} \label{msbl_log_ll}
	\log{p(\matY ; \vgamma)} &=& \sum_{j= 1}^{L} \log{p(\vecy_{j}; \vgamma)} 
	\nonumber \\
	& \propto & - L \log{\left| \matSigma_{\vgamma} \right|} - 
	\text{Tr}\lb \matSigma_{\vgamma}^{-1} \matY \matY^{T} \rb,  
\end{eqnarray}
where $\matSigma_{\vgamma} = \sigma^{2} \matI_{m} + \matA 
\matGamma \matA^{T}$. The log likelihood $\log{p(\matY ; \vgamma)}$ in~\eqref{msbl_log_ll} is 
a nonconvex function of $\vgamma$ and its global maximizer $\hat{\vgamma}$ 
\textcolor{black}{is not available} in closed form. However, its local maximizers can still be found via fixed point iterations or the Expectation-Maximization (EM) procedure. 
In \cite{Wipf_07_msbl}, it is empirically shown that the EM procedure 
faithfully recovers the true support $\setS^{*}$, 
provided $m$ and $L$ are sufficiently large. 



\ifdefined \SKIPTEMP
\subsection{\textcolor{black}{An Interesting Property of Solution of the M-SBL optimization}}
\textcolor{black}{We now comment on $\hat{\vgamma}$, the solution of M-SBL optimization in \eqref{ml_estimate_of_vgamma}. 
Due to the elementwise nonnegativity constraint on the hyperparameter $\vgamma$ in M-SBL, it follows 
from the complementary slackness Karush-Kuhn-Tucker (KKT) condition that $\hat{\vgamma}$ must necessarily satisfy: 
\begin{equation}
	\frac{\partial \log p(\matY; \hat{\vgamma}) }{\partial \vgamma(i)} = 0, \text{ if } \hat{\vgamma}(i) > 0, i \in [n]. 
	\label{eqn_msbl_kkt_implication1}
\end{equation}
From $\eqref{eqn_msbl_kkt_implication1}$, it follows that if $\hat{\vgamma}$ has a nonzero support $\hat{\setS}$, then it must 
satisfy
\begin{equation}
	\veca_{i}^{T} \matSigma_{\hat{\vgamma}}^{-1}
	\lb \matSigma_{\hat{\vgamma}} - \matR_{\vecy\vecy} \rb
	\matSigma_{\hat{\vgamma}}^{-1} \veca_{i} = 0
	\;\;\;\forall\; i \in \hat{\setS}, 
	\label{eqn_msbl_kkt_implication2}
\end{equation}
where $\matSigma_{\hat{\vgamma}} = \sigma^2 \matI_m + \matA \hat{\matGamma} \matA^T$ with $\hat{\matGamma} = \diag(\hat{\vgamma})$. 
If $\hat{\vgamma}$ contains $m$ or more nonzero entries, i.e., $|\hat{\setS}| \ge m$, and since the columns of $\matA$ are taken to be in general position (i.e., $\text{spark}(\matA) = m+1$) and $\matSigma_{\hat{\vgamma}} \in \setS_{++}^{m}$, there exist $m$ distinct linearly independent vectors, namely $\lc \matSigma_{\hat{\vgamma}}^{-1} \veca_i \rc_{i \in \hat{\setS}}$ that together span the nullspace of the Hermitian difference matrix $\matSigma_{\hat{\vgamma}} - \matR_{\vecy\vecy}$. 
Thus, for $|\hat{\setS}| \ge m$, the difference $\matSigma_{\hat{\vgamma}} - \matR_{\vecy\vecy}$ has its nullspace spanning the entire $\Real^m$. Consequently, the column-space of $\matSigma_{\hat{\vgamma}} - \matR_{\vecy\vecy}$ is the trivial set $\{\mathbf{0}_{m}\}$, or equivalently, any M-SBL solution $\hat{\vgamma}$ 
with nonzero support size equal to $m$ or higher must necessarily satisfy
\begin{equation}
\label{eqn_msbl_kkt_consequence}
\matSigma_{\hat{\vgamma}} - \matR_{\vecy\vecy} = \mathbf{0}_{m \times m}. 
\end{equation}
Corollary~\ref{corr_dense_gamma_prop} follows as an immediate implication of the conditional relation in~\eqref{eqn_msbl_kkt_consequence}.} 
\textcolor{black}{
\begin{corollary} \label{corr_dense_gamma_prop}
If the solution of M-SBL optimization in $\eqref{ml_estimate_of_vgamma}$, denoted by $\hat{\vgamma}$, contains $m$ or more nonzero entries, then 
\begin{equation}
\| \hat{\vgamma} \|_{1} \le \frac{m \lvvv \matR_{\vecy \vecy} \rvvv_{2}}{(1 - \alpha)},
\nonumber 
\end{equation}
where $\| \veca_{i} \|_{2}^{2} \ge 1 - \alpha, \forall i \in [n]$.
\end{corollary}
\begin{proof}
From \eqref{eqn_msbl_kkt_consequence}, for $\|\hat{\vgamma}\|_{0} \ge m $, we have 
\begin{align}
\lvvv \matR_{\vecy \vecy} \rvvv_{2} &= \lvvv \matSigma_{\hat{\vgamma}} \rvvv_{2} 
\ge \lvvv \matA \diag(\hat{\vgamma}) \matA^{T} \rvvv_{2}
\ge \frac{\text{tr}(\matA \diag(\hat{\vgamma}) \matA^{T})}{m}
\nonumber \\
& = \frac{1}{m}\sum_{i \in \text{supp}(\hat{\vgamma})} \hat{\vgamma}(i) \| \veca_{i} \|_{2}^{2} 
\ge \frac{\|\vgamma \|_{1}(1 - \alpha) }{m}. 
\label{eqn_intr38}
\end{align} 
By rearranging terms in the final inequality in \eqref{eqn_intr38}, we obtain the desired corollary.
\end{proof}
}
\fi

\subsection{The M-SBL objective is a Bregman matrix divergence} \label{sec:bregman_div_interp}
We now present an interesting interpretation of M-SBL's log-marginalized likelihood objective in \eqref{msbl_log_ll} which 
facilitates a deeper understanding of what is accomplished by its maximization. 
We begin by introducing the Bregman matrix divergence $\D_{\varphi}(\matX, \matY)$ between 
any two $n \times n$ positive definite matrices $\matX$ and $\matY$ as 
\begin{equation} \label{matrix_bregman_divergence}
 \D_{\varphi}(\matX, \matY) \triangleq \varphi(\matX) - \varphi(\matY) - \langle \nabla \varphi(\matY), \matX - \matY \rangle, 
\end{equation}
where $\varphi: S_{++}^{n} \to \Real$ is a convex function with $\nabla \varphi(\matY)$ as its first order 
derivative evaluated at $\matY$. In \eqref{matrix_bregman_divergence}, the matrix inner product 
$\langle \matX, \matY\rangle$ is evaluated as 
$\text{tr} \lb \matX \matY^{T} \rb $.
For the specific case of $\varphi(\cdot) = - \log{|\cdot|}$, 
a strongly convex function, we obtain the Bregman LogDet matrix 
divergence given by
\begin{equation} \label{logdet_matrix_divergence}
 \D_{\text{logdet}}(\matX, \matY) = \text{tr}\lb \matX \matY^{-1} \rb - \log{ \left| \matX \matY^{-1} \right| - n}. 
\end{equation}
By termwise comparison of \eqref{msbl_log_ll} and \eqref{logdet_matrix_divergence}, we observe that the negative log likelihood 
$-\log{p(\matY; \vgamma)}\!$ and $\D_{\text{logdet}}(\matR_{\matY}, \matSigma_{\vgamma})$ are the same up to a constant. 
In fact, in \cite[Theorem~$6$]{ArindamBanerjee05ClusteringWithBregman}, it is shown that there is a one-to-one correspondence between every regular exponential family of probability distributions and a unique and distinct Bregman divergence.

In the divergence term $\D_{\text{logdet}}(\matR_{\matY}, \matSigma_{\vgamma})$, the first argument 
$\matR_{\matY} \triangleq \frac{1}{L} \matY \matY^{T}$ is the sample covariance matrix 
of the observations $\matY$ and the second argument $\matSigma_{\vgamma} = \sigma^{2}\matI + \matA \matGamma \matA^{T}$ is the parameterized covariance matrix of $\matY$. 
This connection between M-SBL's log likelihood cost and the LogDet 
divergence reveals that by maximizing the M-SBL cost, we seek a $\vgamma$ \textcolor{black}{that} minimizes 
the distance between $\matR_{\matY}$ and $\matSigma_{\vgamma}$, 
with point wise distances measured using the Bregman LogDet divergence. Thus, the M-SBL algorithm, at its core, is essentially a 
\textit{second order moment matching} or \textit{covariance matching} procedure for \textcolor{black}{finding} $\vgamma$ such that the associated covariance matrix $\matSigma_{\vgamma}$ 
is closest to the sample covariance matrix, in the Bregman LogDet divergence sense. 

This new interpretation of the M-SBL cost as a Bregman matrix divergence elicits two interesting questions:
\begin{enumerate}[label=\roman*]
 \item Are there other matrix divergences besides LogDet 
	 Bregman matrix divergence which are better suited 
	 for covariance matching?
 \item How to exploit the structural similarities between the M-SBL cost and 
	the Bregman (LogDet) matrix divergence to devise faster and more robust 
	techniques for the type-II likelihood maximization? 
\end{enumerate}
It is our opinion that exploring the use of other matrix divergences for 
covariance matching is  
worth further investigation with the potential for new, improved algorithms for 
support recovery. Preliminary results in this direction have been quite 
encouraging. For example, in \cite{Khanna17RDCMP}, an $\alpha$-R\'{e}nyi 
divergence objective is considered for covariance matching, and a fast greedy algorithm is developed for joint-sparse support recovery. 
\ifdefined \VERBOSE
More details about the connection between the exponential family of distributions 
and Bregman divergences can be found in \cite{DhillonTropp08MatrixNearness}. 
\fi

\section{Some Preliminary Concepts} \label{sec:prelim_concepts}
In this section, we review a few key definitions and results 
which will be used in the later sections.

\ifdefined \SKIPTEMP
\subsection{Restricted Isometry Property}
A matrix $\matA \in \Real^{m \times n}$ is said to satisfy the restricted isometry property (RIP) of order $k$ if there 
exists a constant $\delta^{\matA}_{k} \in \lb 0, 1\rb$ such that 
\begin{equation} \label{defn_rip}
 (1 - \delta_{k}^{\matA}) \lvv \vecx \rvv_{2}^{2} \le 
 \lvv \matA \vecx \rvv_{2}^{2} \le
  (1 + \delta_{k}^{\matA}) \lvv \vecx \rvv_{2}^{2} 
\end{equation}
holds for any $k$-sparse vector $\vecx \in \Real^{n}$.
The smallest such $\delta_{k}^{\matA}$ is called the $k^{\text{th}}$ order restricted isometry constant (RIC) of~$\matA$. 
\fi

\subsection{$\epsilon$-Cover, $\epsilon$-net and covering number}
Suppose $\mathcal{T}$ is a set equipped with a pseudo-metric $d$. 
For any set $\mathcal{A} \subseteq \mathcal{T}$, its $\epsilon$-cover is defined as the coverage 
of $\mathcal{A}$ with open balls of radius $\epsilon$ and centers in $\mathcal{T}$. 
The set  $\mathcal{A}^{\epsilon}$ comprising the centers of these covering balls is called an $\epsilon$-net of $\mathcal{A}$. 
The minimum number of $\epsilon$-balls \textcolor{black}{that} can cover $\mathcal{A}$ is called the $\epsilon$-covering 
number of $\mathcal{A}$, and is given by
\begin{equation} \label{defn_covering_num}
N_{\text{cov}}^{\epsilon} \lb \mathcal{A}, d \rb = \min \lc \lv \mathcal{A}^{\epsilon} \rv: \mathcal{A}^{\epsilon} \text{ is an } \epsilon\text{-net of }\mathcal{A} \rc.
\nonumber
\end{equation}
In computational theory of learning, $\epsilon$-net constructs are often useful in converting a union 
over the elements of a continuous set to a finite sized union.

\ifdefined \SKIPTEMP
\begin{proposition}[\cite{Vershynin09RoleOfSparsity}] \label{prop_covering_num_unit_ball}
Let $\mathcal{B}(0, 1)$ be a unit ball in $\Real^{n}$ centered at $0$. Then, its $\epsilon$-covering number with respect 
to the standard Euclidean metric is bounded as 
$N^{\epsilon}_{\text{cov}}\lb \mathcal{B}(0,1), \lvv \cdot  \rvv_{2} \rb \le \lb 3/\epsilon \rb^{n}$.
\end{proposition}
\fi

\subsection{R\'{e}nyi divergence}
Let $\lb \mathcal{X}, \mathcal{F} \rb$ be a measurable space and $P$ and $Q$ be two probability 
measures on $\mathcal{F}$ with densities $p$ and $q$, respectively, with respect to the dominating  
Lebesgue measure $\mu$ on $\mathcal{F}$. 
Then, for $\alpha \in \Real^{+} \backslash {1}$, the \textit{R\'{e}nyi divergence} of order $\alpha$ between $P$ and $Q$, denoted $\DD_{\alpha}(p || q)$, is defined as 
\begin{equation} \label{defn_renyi_div}
\DD_{\alpha}(p || q) = \frac{1}{\alpha -1} \log \int_{\mathcal{X}} p(x)^{\alpha} q(x)^{1 - \alpha}  \mu(dx).
\nonumber 
\end{equation}
$\DD_{\alpha}(p || q)$ is a nondecreasing function of $\alpha$. For $\alpha \in [0, 1)$, 
$\DD_{\alpha}(p||q) < \DD_{\text{KL}}(p||q)$, with $\displaystyle \lim_{\alpha \to 1} $ $ \DD_{\alpha}(p || q)$ $ = \DD_{\text{KL}}(p || q)$, where 
$\DD_{\text{KL}}$ is the Kullback-Leibler divergence~\cite{Erven14RenyiDivKLDiv}. 
For $p = \mathcal{N}(0, \matSigma_{1})$ and $q = \mathcal{N}(0, \matSigma_{2})$, the $\alpha$-R\'{e}nyi divergence $\DD_{\alpha}(p||q)$ is available in closed form \cite{Gil13RenyiDivClosedForm}, 
\begin{equation} \label{renyi_div_multivar_gauss}
 \DD_{\alpha}(p||q)  = \frac{1}{2(1 - \alpha)} \log{ \frac{\lv (1- \alpha) \matSigma_{1} + \alpha \matSigma_{2} \rv}{ \lv \matSigma_{1} \rv^{1-\alpha}  \lv \matSigma_{2} \rv^{\alpha} }}.
\end{equation}
Proposition~\ref{prop_renyi_div_lb_analytical} provides a lower bound for the $\frac{1}{2}$-R\'{e}nyi divergence between two multivariate Gaussian distributions.
\begin{proposition}\label{prop_renyi_div_lb_analytical}
Let $p_{1}$ and $p_{2}$ be two multivariate Gaussian distributions with 
zero mean and positive definite covariance matrices $\matSigma_{1}$ and 
$\matSigma_{2}$, respectively. Then, the $\frac{1}{2}$-R\'{e}nyi divergence 
between $p_{1}$ and $p_{2}$ is bounded as  
\begin{eqnarray}\label{renyi_div_lb_analytical}
\DD_{\frac{1}{2}}(p_{1}, p_{2}) \ge 
\frac{1}{2} \text{tr}\lb \lb \matSigma_{1} - \matSigma_{2} \rb \lb \matSigma_{1} + \matSigma_{2} \rb^{-1} \lb \matSigma_{1} - \matSigma_{2}\rb \lb \matSigma_{1} + \matSigma_{2}\rb^{-1} \rb.
\nonumber
\end{eqnarray}
\end{proposition}
\begin{proof}
See Appendix \ref{app:proof_thm_renyi_div_lb_analytical}.
\end{proof}

\subsection{Concentration of sample covariance matrix}
\begin{proposition}[Vershynin \cite{Vershynin11RMT}]  \label{prop_conc_spectralnorm_covmat}
Let $\vecyall \in \Real^{m}$ be $L$ independent samples from $\mathcal{N}(0, \matSigma)$, 
and let $\matSigma_{L} = \frac{1}{L}\sum_{j=1}^{L}\vecy_{j}\vecy_{j}^{T}$ denote the sample covariance matrix. 
Then, for any $\epsilon > 0$, 
\begin{equation} \label{eqn_covariance_conc}
\lvvv \matSigma_{L} - \matSigma \rvvv_{2} \le \epsilon \lvvv \matSigma \rvvv_{2}, 
\end{equation}
holds with probability exceeding $1-\delta$ provided $L \ge \frac{C}{\epsilon^2}\log{\frac{2}{\delta}}$, $C$ being an absolute constant.
\end{proposition}

\subsection{Spectral norm bounds for Gaussian matrices} 
\begin{proposition}[Corollary $5.35$ in \cite{Vershynin11RMT}] \label{prop_gaussian_spectral_bound}
Let $\matA$ be an $m \times n$ matrix whose entries are i.i.d. $\mathcal{N}(0,1)$. Then for every $t \ge 0$, 
\begin{equation}
\lvvv \matA \rvvv_{2} \le \sqrt{m} + \sqrt{n} + t,
\nonumber
\end{equation}
with probability at least $1 - 2e^{-t^2/2}$.
\end{proposition}

The following corollary presents a probabilistic bound for the spectral norm of a submatrix of 
an $m \times n$ sized Gaussian matrix obtained by sampling its columns. 
\begin{corollary} \label{corr_Gaussian_submatrix_spectral_bound}
Let $\matA$ be an $m \times n$ sized matrix whose entries are i.i.d. $\mathcal{N}(0,1)$. Then, for any $\setS \subseteq [n], |\setS| \le k$, the submatrix $\matA_{\setS}$ 
obtained by sampling the columns of $\matA$ indexed by $\setS$ satisfies 
\begin{equation}
\lvvv \matA_{\setS} \rvvv_{2} \le \sqrt{m} + \sqrt{k} + \sqrt{6 k \log n}    
\nonumber.
\end{equation}
with probability exceeding $1 - 2 n^{-k}$. 
\end{corollary}
\begin{proof}
See Appendix~\ref{app:proof_corr_Gaussian_submatrix_spectral_bound}.
\end{proof}

\section{Support Error Analysis} 
\label{sec:abstract_perr_bound_derivation}
Towards studying the sparsistency of the M-SBL solution, we first derive a 
Chernoff bound for the support error probability incurred by the solution of a 
constrained version of M-SBL optimization under assumption~\textbf{A1}. We 
begin by introducing some of the frequently used notation in the table below.
\begin{table}[h]
\centering
\renewcommand{\arraystretch}{1.2}
\begin{tabular}{|l|p{12cm}|}
\hline
$\setS^{*}$ & True row support of $\matX$. \\ \hline
$\vgamma^{*}$ & Principal diagonal of the common covariance matrix  $\matGamma^*$ of   
the i.i.d. columns in $\matX$. Consequently, $\text{supp}(\vgamma^{*}) = \setS^*$.\\
\hline
$K$ & Maximum number of nonzero rows in $\matX$\\ \hline 
$\setS_{k}$ & The collection of all support sets of $k$ or lesser 
 size, i.e., $\setS_{k} = \lc \setS \subseteq [n], |\setS| \le k \rc$ \\ \hline
$\Theta(\setS)$ & Bounded hyperparameter set associated 
with the support set $\setS$, formally defined as $ \Theta(\setS) \! \triangleq \!
 \lc 
 \vgamma \in \Real^{n}_{+}: \text{supp}(\vgamma) = \setS,
 \vgammamin \preceq \vgamma_{\setS} \preceq \vgammamax 
 \rc$.\\ \hline
$\Theta_{k}$ & The collection of all $k$ or less sparse vectors in $\Real^{n}_{+}$  
with nonzero coefficients lying in $[\vgammamin, \vgammamax]$. 
By definition, we have $\Theta_{k} = \bigcup_{\setS \in \setS_{k}}\Theta(\setS)$.\\ \hline
\end{tabular}
\end{table}

By assumption \textbf{A1} on $\matX$, $\vgamma^{*}$ belongs to the bounded parameter set $\Theta_{n}$.
Therefore, in order to estimate $\vgamma^{*}$ from the measurement vectors 
$\matY$, we consider solving a constrained variant of the M-SBL optimization in 
\eqref{ml_estimate_of_vgamma}, which we refer to as the cM-SBL problem:
\begin{equation} \label{defn_cmsbl}
	\hspace{-1cm} \text{cM-SBL:} \hspace{0.5cm}
	\hat{\vgamma} = \underset{\vgamma \in \Theta_{n}}
	{\text{arg max }} \LL(\matY; \vgamma).
\end{equation}
The cM-SBL objective $\LL(\matY; \vgamma)$ is the same as the M-SBL's log-marginalized likelihood $\;$ $\log p(\matY; \vgamma)$ defined 
in~\eqref{msbl_log_ll}. 
The row support of $\matX$ is estimated as the support of $\hat{\vgamma}$, where $\hat{\vgamma}$ is a solution of \eqref{defn_cmsbl}.
Consider the set of ``bad" MMVs,  
\begin{equation} 
\mathcal{E}_{\setS^{*}} \triangleq \lc \matY \in \Real^{m \times L} : \text{supp}\lb \hat{\vgamma} \rb \neq \setS^{*} \rc, 
\label{bad_mmv_set}
\end{equation}
which result in erroneous estimation of $\setS^*$ via 
cM-SBL. In other words, $\mathcal{E}_{\setS^{*}}$ is the collection of 
undesired MMVs for which the M-SBL objective is globally maximized by some 
$\vgamma \in \Theta(\setS)$, $\setS \neq \setS^*$, i.e., 
\begin{equation}\label{error_event_as_union}
	\mathcal{E}_{\setS^{*}} = \!\!\!\!
	\bigcup_{\setS \in \setS_{n} \backslash \lc \setS^{*} \rc} 	\lc \matY : \!\!
	\max_{\vgamma \in \Theta(\setS)}
	\LL(\matY; \vgamma)
	\ge \!
	\max_{\vgamma^{\prime} \in \Theta(\setS^{*})} \LL(\matY; \vgamma^{\prime})
	\rc.
\end{equation}
We are interested in finding the conditions under which $\PP(\mathcal{E}_{\setS^*})$ can be made arbitrarily small.
Since $\displaystyle \max_{\vgamma^{\prime} \in \Theta(\setS^{*})}
\LL(\matY; \vgamma^{\prime}) \ge \LL(\matY; \vgamma^{*})$, it follows that
\begin{eqnarray}
	&&\hspace{-1.1cm} \mathcal{E}_{\setS^{*}} \subseteq \!\!\!
	\bigcup_{\setS \in \setS_{n} \backslash \lc \setS^{*} \rc} 	  \lc \matY :
	\max_{\vgamma \in \Theta(\setS)}
	\LL(\matY; \vgamma)
	\ge \LL(\matY; \vgamma^{*}) \rc
	\nonumber \\
	&& \hspace{-0.5cm} = \!\!\!
	\bigcup_{\setS \in \setS_{n} \backslash \lc \setS^{*} \rc} 	  
	\bigcup_{\vgamma \in \Theta(\setS)} \!\!
	\lc \matY :
	\LL(\matY; \vgamma) - \LL(\matY; \vgamma^{*}) \ge 0 \rc.
	\label{err_evt_as_continuous_union}
\end{eqnarray}
The continuous union over infinitely many elements of $\Theta(\setS)$ in \eqref{err_evt_as_continuous_union} can be relaxed to a finite sized union by 
using the following $\epsilon$-net argument.
Consider $\Theta^{\epsilon}(\setS)$, a finite sized $\epsilon$-net of the 
hyperparameter set $\Theta(\setS)$, 
such that for any $\vgamma \in \Theta(\setS)$, there exists an element $\vgamma^{\prime} \in \Theta^{\epsilon}(\setS)$ satisfying $\lv \LL(\matY; \vgamma) - \LL(\matY; \vgamma^{\prime}) \rv \le \epsilon$. 
Proposition~\ref{prop_cov_num_enet} gives an upper bound on the size of such an $\epsilon$-net. 
\begin{proposition} \label{prop_cov_num_enet}
Given a support set $\setS \subseteq [n]$, there exists a finite set 
$\Theta^{\epsilon}(\setS) \subset \Theta(\setS)$ such that it simultaneously satisfies
\begin{enumerate}[label=(\roman*)]
	\item For any $\vgamma \in \Theta(\setS)$, there exists 
		a $\vgamma^{\prime} \in \Theta^{\epsilon}(\setS)$ 
		such that $\lv \LL(\matY; \vgamma) - \LL(\matY; \vgamma^{\prime}) \rv \le \epsilon$.  
	\item $ \lv \Theta^{\epsilon}(\setS) \rv \le
	    \max \lc 1,
		\lb 3 C_{\LL, \setS} (\vgammamax - \vgammamin) \sqrt{|\setS|} / \epsilon \rb^{|\setS|} \rc$, 
		where $C_{\LL, \setS}$ is the Lipschitz constant of $\LL(\matY; \vgamma)$ with respect to $\vgamma$ in the bounded domain $\Theta(\setS)$.
\end{enumerate}
The set $\Theta^{\epsilon}(\setS)$ is an $\epsilon$-net of $\Theta(\setS)$.
\end{proposition}
\begin{proof}
	 See Appendix \ref{app:proof_prop_cov_num_enet}.
\end{proof}
From Proposition~\ref{prop_cov_num_enet}-(ii), we observe that both the construction as well as size of $\Theta^{\epsilon}(\setS)$ depends on the Lipschitz continuity 
of the log-likelihood $\LL(\matY; \vgamma)$ with respect to $\vgamma$. 
By virtue of data-dependent nature of $\LL(\matY; \vgamma)$, its Lipschitz constant 
$\mathcal{C}_{\LL, \setS}$ depends on the instantaneous value of $\matY$. 
To make the rest of the analysis independent of~$\matY$, we introduce a new MMV set $\setG$, conditioned on 
which, the Lipschitz constant~$\mathcal{C}_{\LL, \setS}$ is uniformly bounded solely in terms of second-order statistics of $\matY$. A possible choice of $\setG$ could be   
\begin{equation}
\setG \triangleq \lc \matY  \subset \Real^{m \times L} : \lvvv \frac{1}{L} 
\matY \matY^{T} \rvvv_{2} \le 
2 \lvvv \mathbb{E} \ls \vecy_{1} \vecy_{1}^{T} \rs \rvvv_{2} \rc.
\label{defn_bdd_spectral_norm_mmv_set}
\end{equation}
By Proposition~\ref{loglikelihood_Lipschitz_char} in Appendix~\ref{proof_prop_kcov_ub}, for $\matY \in \setG$, $\LL(\matY; \vgamma)$ is uniformly continuous with a Lipschitz constant that depends only 
on the spectral norm of $\mathbb{E}[\vecy_{1} \vecy_{1}^{T}]$. Hence, the $\epsilon$-net can now be constructed entirely independent of $\matY$. We denote this $\epsilon$-net by $\Theta^{\epsilon}(\setS)\vert_{\setG}$.

Since for arbitrary sets $\setA$ and $\setB$, 
$\setA \subseteq (\setA \cap \setB) \cup \setB^{c}$, the RHS in \eqref{err_evt_as_continuous_union} relaxes as 
\begin{equation} 
\mathcal{E}_{\setS^{*}} \!\subseteq \!
 	\lc \!\!
 	\bigcup_{\substack{\setS \in  \setS_{n} \backslash \setS^{*} }} 
 	\bigcup_{\substack{\vgamma \in \Theta(\setS)}}  \!\!\!\!
 	 \lc \! \LL(\matY; \vgamma) - \LL(\matY; \vgamma^{*}) \ge 0 \rc 
 	 \cap \setG  \! \rc 
 	 \cup \; \setG^{c}. 
\end{equation}
The continuous union over $\Theta(\setS)$ relaxes to a finite sized union over $\Theta^{\epsilon}(\setS)\vert_{\setG}$ as shown below.
\begin{eqnarray} 
\mathcal{E}_{\setS^{*}} 
& & \hspace{-0.7cm} \subseteq\!\!
\lc \!\!
\bigcup_{\substack{\setS \in \setS_{n} \backslash \setS^{*} }} \!
 \bigcup_{\substack{\vgamma \in \Theta^{\epsilon}(\setS)\vert_{\setG}}} \!\!\!\!
 \lc \LL(\matY; \vgamma) \! -\! \LL(\matY; \vgamma^{*}) \! \ge\! -\epsilon \rc 
 \cap \setG  \!\rc \!
 \cup  \setG^{c}
 \nonumber \\
&& \hspace{-0.8cm} \subseteq \!
\lc \!
\bigcup_{\substack{\setS \in \setS_{n} \backslash \setS^{*} }} \;
 \bigcup_{\substack{\vgamma \in \Theta^{\epsilon}(\setS)\vert_{\setG}}} \!\!\!\!\!
 \lc \LL(\matY; \vgamma) - \LL(\matY; \vgamma^{*}) \ge -\epsilon \rc \rc \!
 \cup  \setG^{c}. 
 \nonumber
\end{eqnarray}
By applying the union bound, we obtain
\begin{align}
 \PP\lb \mathcal{E}_{\setS^{*}} \rb  & \le  
 \sum_{\substack{\setS \in \setS_{n} \backslash \setS^{*} }} 
 \sum_{\substack{\vgamma \in \Theta^{\epsilon}(\setS)\vert_{\setG}}}
\PP \lb \LL(\matY; \vgamma) - \LL(\matY; \vgamma^{*}) \ge -\epsilon \rb 
 +  \PP \lb \setG^{c} \rb. 
\label{err_prob_ub_as_double_sum} 
\end{align} 
From \eqref{err_prob_ub_as_double_sum}, the support error probability $  \PP(\mathcal{E}_{\setS^{*}})$ will be small when the summands, $\PP \lb \LL(\matY; \vgamma) \right. $ $ \left. - \LL(\matY; \vgamma^{*}) \ge -\epsilon \rb$, 
$\vgamma \in \Theta^{\epsilon}(\setS) \vert_{\setG} $, are individually sufficiently small so that their collective contribution remains small, and $\PP \lb \setG^{c} \rb$ is also small.
In Theorem \ref{thm_model_mismatch_ldp}, we show that each summand corresponds to a large deviation event which occurs with an exponentially decaying probability.

\begin{theorem} \label{thm_model_mismatch_ldp}
For $\vgamma \in \Real^{n}_{+}$, let $p_{\vgamma}$ denote the marginal probability density of the columns of $\matY$ 
induced by the joint-sparse columns of $\matX$ drawn independently from 
$\mathcal{N}\lb 0, \diag{(\vgamma)} \rb$. Then, the 
log-likelihood  $\LL(\matY; \vgamma) = \sum_{j = 1}^{L}\log{p_{\vgamma}(\vecy_{j})}$ satisfies 
the following large deviation property.
\begin{equation} \label{eqn_model_mismatch_ldp}
 \PP \lb \LL(\matY; \vgamma) - \LL(\matY; \vgamma^{*}) 
 \ge - 
 \epsilon \rb 
 \le 
 \exp{\lb - L \psi^* \lb - \frac{\epsilon}{L} \rb \rb},
\end{equation}
where $\psi^{*}(\cdot)$ is the Legendre transform\footnote
{
For any convex function $f: \mathcal{X} \to \Real$ on a 
convex set $\mathcal{X}\subseteq \Real^{n}$, its 
Legendre transform is the function 
$f^{*}$ defined by 
\begin{equation} \label{defn_legendre_transform}
	f^{*}(\vecz) = \sup_{\vecx \in \mathcal{X}}  
	\lb \langle \vecz, \vecx \rangle  - f(\vecx)\rb.
	\nonumber
\end{equation}
}
of $\psi(t) \triangleq (t-1)\DD_{t}(p_{\vgamma}, p_{\vgamma^*})$, 
and $\DD_{t}$ is the $t$-R\'{e}nyi divergence (of order 
$t > 0$) between the probability densities $p_{\vgamma}$ and 
$p_{\vgamma^{*}}$.
\end{theorem}
\begin{proof}
	See Appendix \ref{app:thmproof_err_evt_as_ldp}.
\end{proof}

Note that, when the measurement noise is Gaussian, the marginal density $p_{\vgamma}(\vecy_{j})$ of the individual observations $\vecy_{j}$ is also Gaussian with zero mean and covariance matrix $\matSigma_{\vgamma} = \sigma^{2} \matI_{m} + \matA \matGamma \matA^{T}$. If $\sigma^{2} > 0$, both marginals 
$p_{\vgamma}$ and $p_{\vgamma^{*}}$ are non-degenerate and hence the R\'{e}nyi divergence  
$\DD_{t}(p_{\vgamma}, p_{\vgamma^*})$ in Theorem~\ref{thm_model_mismatch_ldp} is 
well defined. We now restate Theorem \ref{thm_model_mismatch_ldp} as  
Corollary~\ref{coro_model_mismatch_ldp}, 
which is the final form of the large deviation result for~$\LL(\matY; \vgamma)$ used later for bounding $\PP(\mathcal{E}_{\setS^{*}})$.
\begin{corollary} \label{coro_model_mismatch_ldp}
For any $\vgamma \in \Real^{n}_{+}$, and the true variance parameters $\vgamma^{*}$, let the 
associated marginal densities $p_{\vgamma}$ and~$p_{\vgamma^*}$ be defined as in Theorem 
\ref{thm_model_mismatch_ldp}, and suppose $\sigma^2 > 0$. Then, the log-likelihood $\LL(\matY; \vgamma)$ satisfies the large deviation property
\begin{equation}
\PP \lb \!\LL(\matY;\! \vgamma) \!-\! \LL(\matY;\! \vgamma^{*}) 
\ge \!
- \frac{L \DD_{\frac{1}{2}}(p_{\vgamma}, p_{\vgamma^{*}}\!)}{2} \!\rb  
\le\!
e^{ 
- \frac{L {\DD_{\frac{1}{2}} (p_{\vgamma}, p_{\vgamma^{*}} )}}{4} }.
\label{large_dev_event}
\end{equation}
\end{corollary}
\begin{proof}
The large deviation result is obtained by replacing $ \psi^* \lb - \frac{\epsilon}{L} \rb$ in 
Theorem \ref{thm_model_mismatch_ldp} by its lower bound 
$- \frac{t \epsilon}{L} - \psi \lb t \rb$, 
followed by setting $t = 1/2$ and $\epsilon = L \DD_{\frac{1}{2}}(p_{\vgamma} , p_{\vgamma^{*}})/2$.
\end{proof}
Note that, in the above, we have used the suboptimal choice $t=1/2$ for the Chernoff parameter~$t$, since its optimal value is not available in closed form. However, this suboptimal selection of $t$ is inconsequential as it figures only as a multiplicative factor in the final MMV complexity. By using Corollary \ref{coro_model_mismatch_ldp} in \eqref{err_prob_ub_as_double_sum}, we can bound 
$\PP \lb \mathcal{E}_{\setS^{*}}\rb$ as  
\begin{equation} \label{perr_err_exp_form1}
\PP(\mathcal{E}_{\setS^{*}}) \le 
\sum_{\setS \in \setS_{n} \backslash \setS^{*}} \!\!\!\!
\lv \Theta^{\epsilon}(\setS) \vert_{\setG} \rv \exp{\lb - \frac{L \DD_{\setS}^{*}}{4} \rb} \;+\; \PP\lb \setG^{c}\rb,
\end{equation}
with $\epsilon = \frac{L \DD^{*}_{\setS}}{2}$ and $\DD^{*}_{\setS}$ defined as 
\begin{equation}
\DD^{*}_{\setS} \triangleq \inf_{\vgamma \in \Theta(\setS)} \DD_{\frac{1}{2}} \lb p_{\vgamma}, p_{\vgamma^*} \rb.
\label{eqn_minD}
\end{equation}
Suppose the support $\setS$ differs from $\setS^{*}$ in exactly $k_{d}^{\setS, \setS^*}$ locations, then 
\begin{equation} 
\PP(\mathcal{E}_{\setS^{*}})  \le  \sum_{\setS \in \setS_{n} \backslash \setS^{*}} 
\exp{\lb  - L k_{d}^{\setS, \setS^{*}} \lb \frac{\eta}{4} - \frac{\kappa_{\text{cov}}}{L} \rb\rb} + \PP\lb \setG^{c}\rb,
\label{perr_form2}
\end{equation}
where
\begin{eqnarray} 
\label{eta_char}
\eta & \triangleq &  \min_{\setS \in \setS_{n} \backslash \setS^*}  \frac{\DD_{\setS}^{*}}{k_{d}^{\setS, \setS^*}} , 
\\
\kappa_{\text{cov}} &\triangleq& \max_{\setS \in \setS_{n} \backslash \setS^*}  \frac{ \log \big| \Theta^{\epsilon}(\setS) \vert_{\setG} \big|}{k_{d}^{\setS, \setS^*}} , \; \epsilon = \frac{L \DD_{\setS}^*}{2}.
\label{epsilon_cover_cardinality_bound}
\end{eqnarray}

Using the above, we can state the following theorem.
\begin{theorem} \label{thm_abstract_suff_condition}
Suppose $\setS^{*}$ is the true row support of the unknown $\matX$ satisfying assumption \textbf{A1} and $|\setS^{*}| \le K$. Then, for any $\delta \in (0,1)$, $\PP(\mathcal{E}_{\setS^{*}}) 
\le 2 \delta$, if
\begin{equation} \label{abstract_mmv_bound}
L \ge 
\max  \lc \frac{8}{\eta} \log {\lb \frac{6enK}{\delta} \rb }, 
\frac{8 \kappa_{\text{cov}}}{\eta} , 
C \log{\frac{2}{\delta}} \rc.
\nonumber
\end{equation}
Here, $\eta$ and $\kappa_{\text{cov}}$ are as defined in \eqref{eta_char}
and \eqref{epsilon_cover_cardinality_bound}, respectively, and $C>0$ is a universal numerical constant.
\end{theorem}
\begin{proof}
See Appendix~\ref{app:proof_thm_abstract_suff_condition}.
\end{proof}
In Theorem \ref{thm_abstract_suff_condition}, we finally have an abstract bound on the 
sufficient number of MMVs, $L$, which guarantees vanishing support error 
probability in cM-SBL, given that  
the true support is $\setS^*$. However, the MMV bound is meaningful only when 
$\eta$ \eqref{eta_char} is strictly positive. We now 
proceed to deduce the conditions for which 
\begin{enumerate}
\item $\eta > 0$,
\item $\eta$ and $\kappa_{\text{cov}}$ scale favorably with the 
MMV problem dimensions. 
\end{enumerate}

\subsection{Bounds for $\eta$ and $\kappa_{\text{cov}}$}
To understand how small $\eta$ in the MMV bound in  
Theorem~\ref{thm_abstract_suff_condition} can be, we 
first derive a lower bound on~$\DD_{\setS}^{*}$ for any $\setS \subseteq [n]$, in the following proposition. 
\begin{proposition} \label{renyi_div_lb_gaussian_case} 
Let $p_{\vgamma}$ denote the parameterized multivariate Gaussian density with 
zero mean and covariance matrix $\matSigma_{\vgamma} = \sigma^{2} \matI + \matA 
\matGamma \matA^{T}$, 
$\matGamma = \diag(\vgamma)$. 
For any pair $\vgamma, \vgamma^* \in \Real^{n}_{+}$ such that $\setS = 
\text{supp}(\vgamma)$ and $\setS^* = \text{supp}(\vgamma^*)$, the 
$\frac{1}{2}$-R\'{e}nyi divergence between $p_{\vgamma}$ and 
$p_{\vgamma^*}$ satisfies
\begin{equation}
\DD_{\frac{1}{2}} \lb p_{\vgamma}, p_{\vgamma^*} \rb
\ge 
\frac{ \| (\matA \odot \matA) (\vgamma - \vgamma^*) \|_{2}^{2}}
{4 \lb \sigma^2 + \vgammamax \sigma_{\text{max}}^{2}(\matA_{\setS \cup 
\setS^*}) \rb^{2}},
\nonumber 
\end{equation}
where $\matA \odot \matA$ denotes the columnwise Khatri-Rao product of~$\matA$ 
with itself and  $\sigma_{\text{max}}(\cdot)$ denotes maximum singular value of 
the input matrix.
\end{proposition}
\begin{proof}
See Appendix \ref{app:renyi_div_lb_gaussian_case}.
\end{proof}
From Proposition~\ref{renyi_div_lb_gaussian_case}, it can be observed that as long as the null space of $\matA \odot \matA$ is devoid of any vectors of the form 
$\vgamma - \vgamma^*$ (i.e., difference of a nonnegative vector and a 
nonnegative $K$-sparse vector), then $\eta$ as 
defined in \eqref{eta_char} is always strictly positive. This condition for strictly positive $\eta$ can be formalized as the \textit{nonnegative restricted null space property} of $\matA \odot \matA$, defined next. 
\begin{defn} \label{defn_rnsp}
A matrix is said to satisfy the nonnegative restricted null space property (NN-RNSP) of order $k$ if its null space does not contain any vectors that are expressible as the difference between a $k$ (or lesser) sparse nonnegative vector and an arbitrary nonnegative vector. 
\end{defn}

The requirement that $\matA \odot \matA$ satisfies a 
restricted null-space property similar to Definition~\ref{defn_rnsp} 
has been highlighted in~\cite{QiaoPal19MMV} in the context of MMV-based 
support recovery using correlation-aware priors. In \cite{QiaoPal19MMV}, 
it is shown that $\matA \odot \matA$ exhibits the NN-RNSP for appropriate 
sparse-array designs of the sensing matrix $\matA$ provided the $k$-sparse 
nonnegative vector in Definition~\ref{defn_rnsp} satisfies a certain 
\textit{support-separability condition}. In the theorem below, we present an interesting robust null-space property of $\matA \odot \matA$ which is satisfied 
under the mild condition that the columns of $\matA$ are approximately normalized. This property will be 
crucial in establishing $\matA \odot \matA$'s NN-RNSP compliance towards ensuring the positivity of~$\eta$, and also dispenses with the restrictive support-separability 
condition required in \cite{QiaoPal19MMV}.

\begin{theorem}[\textit{\textbf{Strong Robust Null Space Property of Self Khatri-Rao Products}}] \label{thm_strong_rnsp_kr}
Let $\matA$ be an $m \times n$ sized real matrix with columns $\veca_{i}$ satisfying 
$\lvv \veca_{i} \rvv_{2}^{2} \in \ls 1 - \alpha, 1 + \alpha \rs$ for some $\alpha \in (0,1)$ for $i \in [n]$. 
Then, the self Khatri-Rao product $\matA \odot \matA$ satisfies the following 
robust null space property:
\begin{equation}
\lvv (\matA \odot \matA) \vecv \rvv_{2}^{2} \ge 
\frac{(1 - \alpha)^2}{2m} \lb \lvv \vecv_{+} \rvv_{1}^{2} +  \lvv \vecv_{-} \rvv_{1}^{2} \rb
\nonumber 
\end{equation} 
for all $\vecv \in \Real^{n}$ such that $\lvv \vecv_{+} \rvv_{1} \ge 4 \lb \frac{1 + \alpha}{1-\alpha} \rb^{2} \lvv \vecv_{-} \rvv_{1}$.
Here, $\vecv_{+}$ and $\vecv_{-}$ are nonnegative vectors containing the absolute values of the positive and negative elements of $\vecv$, respectively, such that $\vecv = \vecv_{+} - \vecv_{-}$.
\end{theorem}
\begin{proof}
	See Appendix~\ref{app:proof_thm_strong_rnsp_kr}.
\end{proof}
An interesting consequence of Theorem~\ref{thm_strong_rnsp_kr} is that as long 
as~$\matA$ is approximately column normalized, the null space of $\matA \odot 
\matA$ does not contain any vectors of the form $\vgamma - \vgamma^*$ when 
$\vgamma$ and $\vgamma^*$ have widely different nonzero supports, 
particularly if $\| \vgamma \|_{0} \ge 4 \lb \frac{\vgammamax}{\vgammamin} 
\rb \lb \frac{1 + 
\alpha}{1- \alpha} \rb^{2} \| 
\vgamma^* \|_{0}$.
However, when the support sizes of $\vgamma$ and $\vgamma^*$ are comparable, 
it is not as straightforward to ascertain whether $\vgamma - \vgamma^*$ lies in the 
null space of $\matA \odot \matA$ or not. 
In Proposition~\ref{prop_eta_lb}, we state verifiable sufficient conditions that 
subsume the NN-RNSP condition for $\matA \odot \matA$, which in turn 
guarantees that $\eta$ in \eqref{eta_char} is always strictly positive.
\begin{proposition} \label{prop_eta_lb}
Let the sensing matrix $\matA = \ls \veca_{1}, \veca_{2}, \ldots , \veca_{n} \rs$ satisfy the following two properties\footnote{Properties $P1$ and $P2$ ensure that $\text{Null}(\matA \odot \matA)$ does not contain any vectors 
of the form $\vgamma - \vgamma^*$ when $\| \vgamma \|_{0} \ge K + 
K_{\text{threshold}}$ and when $\| \vgamma \|_{0} < K + K_{\text{threshold}}$, respectively.}
\begin{description}
\item[(P1).] $\exists \alpha \in (0,1)$ such that $\| \veca_{i} \|_{2}^{2} \in \ls 1- \alpha, 1+ \alpha\rs, \forall i \in [n]$.
\item[(P2).] For $K_{\text{threshold}} \triangleq 4 K \left(\frac{\vgammamax}{\vgammamin} \lb \frac{1 + \alpha}{1- \alpha} \rb^{2} \right)$, $\exists \beta > 0$ such that 
$\| (\matA \odot \matA) \vecv \|_{2}^{2} \ge \beta \| \vecv \|_{2}^{2}$
for all $k$-sparse vectors $\vecv \in \Real^{n}$ and $1 \le k \le (K + K_{\text{threshold}})$.
\end{description} 
Then, $\eta$ defined in \eqref{eta_char} is lower bounded as
\begin{align}
\eta &\ge 
\frac{\vgammamin^2}{4 \lb \sigma^2 +  \vgammamax \rb^{2}}  
\min \lb  
\frac{\beta}{\max \lb 1,\delta^{2}_{(K + K_{\text{threshold}})} \rb }, 
 \frac{(1 - \alpha)^2}{4m} \!\!\!\! \min_{\substack{\setS \subseteq [n], \\ |\setS \backslash \setS^*| + 
		|\setS^* \backslash \setS| > K_{\text{threshold}} } }
\!\!\!\frac{|\setS \cup \setS^*|}{\max \lb 1,  \delta^{2}_{|\setS \cup \setS^*|} \rb}
\rb,
\label{eta_lb}
\end{align}
where $ \delta_{k} \triangleq \displaystyle \max_{\setS \subseteq [n]: |\setS| 
\le k} \lvvv 
\matA_{\setS}^{T} \matA_{\setS}\rvvv_{2}$ for any $k \in [n]$.
\end{proposition}
\begin{proof}
	See Appendix \ref{app:proof_prop_eta_lb}.
\end{proof}
By substituting $\eta$ in Proposition~\ref{prop_cov_num_enet} with its lower 
bound in Proposition~\ref{prop_eta_lb}, one can upper bound 
$\kappa_{\text{cov}}$ as follows.
\begin{proposition} \label{prop_kcov_ub}
For the same setting as Proposition~\ref{prop_eta_lb}, 
\begin{equation}
\kappa_{\text{cov}} \le 
 \ls (K + 5) \lb 
\log m + \log K + \log \Delta_{\kappa_{\text{cov}}} \rb
+ 2 \log n
\rs^{+},
\nonumber
\end{equation} 
where 
\begin{align}
\Delta_{\kappa_{\text{cov}}} & \triangleq   
\frac{24}{\vgammamin^2} \lb \frac{\vgammamax}{\vgammamin} - 1 \rb \lb 3 +  
\frac{2\vgammamax}{\sigma^2} \rb 
(\sigma^2 + \vgammamax)^2
\max \lb \frac{ \delta^{2}_{K + 
K_{\text{threshold}}}}{\beta} , \frac{9}{2} \frac{(1+ \alpha)^2}{(1 - \alpha)^{2}} \rb 
\nonumber
\end{align} 
and $[\cdot]^{+} =\max(\cdot, 0)$.
\end{proposition}
\begin{proof}
	See Appendix \ref{proof_prop_kcov_ub}.
\end{proof}
The above bounds for $\eta$ and $\kappa_{\text{cov}}$ are valid for any deterministic or randomly constructed sensing matrix~$\matA$. The validity of these bounds is contingent upon showing the existence of a strictly positive $\beta$ that satisfies condition~$P2$ of Proposition~\ref{prop_eta_lb}, which tantamounts to showing 
that any submatrix of $\matA \odot \matA$ obtained by sampling its $K + K_{\text{threshold}}$ columns is nonsingular. 
The quantity $\beta$ is commonly referred to as the \textit{restricted minimum singular value} of the self Khatri-Rao product $\matA \odot \matA$. 
To illustrate how the derived bounds for $\eta$ and 
$\kappa_{\text{cov}}$ scale with the MMV problem size, we consider two example scenarios in the following corollary wherein~$\matA$ is randomly constructed and the associated $\eta$ is known to be strictly positive. 
\begin{corollary} \label{corr_gaussian_meas_map_eta_char}
Let $\matA$ be an $m \times n$ sized matrix with i.i.d. $\mathcal{N} \lb 0, \frac{1}{m} \rb$ entries. Then, we have 
\begin{enumerate}[label=(\roman*)]
\item $\eta = \Omega \lb \frac{K}{n} \rb$ for $m = \Theta(K \log n)$,  
\item $\eta = \Omega \lb \frac{\sqrt{K}}{n} \rb$ for $m = \Theta(\sqrt{K} \log n)$,  
\end{enumerate}
and $\kappa_{\text{cov}} = O \lb K \log K + K \log \log n + \log n \rb$ in both cases, with probability exceeding $1 - \Theta(n^{-2})$. 
\end{corollary}
\begin{proof}
	See Appendix \ref{app:proof_corr_gaussian_meas_map_eta_char}.
\end{proof}

\ifdefined \SKIP
That is,
\begin{equation} \label{worst_err_exponent}
\DD^{*}_{1/2} \triangleq 
\min_{\vgamma \in \Theta^{\epsilon}(\setS) \vert_{\setG}, \; \setS \in \setS_{K} \backslash \setS^{*}} 
\DD_{1/2}\lb p_{\vgamma}, p_{\vgamma^{*}} \rb.
\end{equation}
We introduce $\eta$ as a strictly positive lower bound of the worst case exponent $\DD^{*}_{1/2}$, 
which is later shown to depend only on $\matA, \sigma^{2}, \vgammamin$ and $\vgammamax$. 
\textcolor{black}{Using $\eta$,} the upper bound for $\PP(\mathcal{E}_{\setS^{*}})$ in~\eqref{perr_err_exp_form1} simplifies to 
\begin{eqnarray} \label{perr_abstract_bound}
\PP(\mathcal{E}_{\setS^{*}})  \!\!\!\! &\le& \!\!\!\!
\lv \setS_{K} \backslash \setS^{*} \rv
\lb \sup_{\setS \in \setS_{K}} \lv \Theta^{\epsilon}(\setS) \vert_{\setG} \rv \rb
\exp{\lb - \eta L / 4 \rb} + \PP \lb \setG^{c}\rb
\nonumber \\
&& \hspace{-1.2cm} \le 
\lv \setS_{K} \backslash \setS^{*} \rv
\lb \sup_{\setS \in \setS_{K}} \lv \Theta^{\frac{\eta L}{2}	}(\setS) \vert_{\setG} \rv \rb
\exp{\lb - \eta L / 4 \rb} + \PP \lb \setG^{c}\rb
\nonumber \\
&& \hspace{-1.2cm} \le 
\lv \setS_{K} \rv \;
\lb \kappa_{\text{cov}} \rb \;
\exp{\lb - \eta L / 4 \rb} + \PP \lb \setG^{c}\rb
\nonumber \\
&& \hspace{-1.2cm} = 
\exp{\lb -L \lb 
\frac{\eta}{4}
- \frac{\log{\lv \setS_{K} \rv}}{L}
- \frac{\log{\kappa_{\text{cov}}  }}{L}
\rb\rb} + \PP \lb \setG^{c}\rb.  
\label{eq:Pe_sstar_ub_eta}
\end{eqnarray}
In the above, the second inequality is a consequence of the fact that 
the size of an $\epsilon$-net of any set decreases monotonically with $\epsilon$, and that $\DD_{1/2}^{*} \ge \eta $. In the third inequality, 
we introduce a new constant $\kappa_{\text{cov}}$ defined as 
\begin{equation} \label{epsilon_cover_cardinality_bound}
\kappa_{\text{cov}} \triangleq \max_{\setS \in \setS_{K}} \big| \Theta^{\epsilon}(\setS) \vert_{\setG} \big| 
\end{equation}
evaluated at $\epsilon = \eta L/2$. 
The exponential bound in~\eqref{perr_abstract_bound} reveals that if $\eta$ is strictly positive, 
then for a sufficiently large number of MMVs, specifically 
$L > \frac{4}{\eta}  \lb \log{\lv \setS_{K} \rv} + \log{\kappa_{\text{cov}}} \rb$, \textcolor{black}{the first bounding term for} 
$\PP(\mathcal{E}_{\setS^{*}})$ decays exponentially with the number of MMVs, 
and therefore can be driven arbitrarily close to zero. 
\textcolor{black}{Further, by Proposition \ref{prop_conc_spectralnorm_covmat}, $\PP(\setG^{c}) \le \delta/2$ for $L \ge C \log{(4/\delta)}$, where 
$C > 0$ is a numerical constant.}
In fact, for any $\delta > 0$, $\PP(\mathcal{E}_{\setS^{*}}) \le \delta$ is guaranteed if 
\begin{equation} \label{abstract_mmv_bound_probabilistic}
 L \ge \max \lb \frac{4}{\eta}  \lb \log{\lv \setS_{K} \rv} + \log{\kappa_{\text{cov}}} + \log{\frac{2}{\delta}} \rb , 
 C \log{\frac{4}{\delta}}\rb,  
\end{equation}
provided $\eta$ is strictly positive.
Thus, in \eqref{abstract_mmv_bound_probabilistic}, we have an abstract bound on the number of MMVs which guarantees 
an arbitrarily small probability of erroneous support recovery by cM-SBL($K$), given that the true support is $\setS^{*}$.

The condition $\eta >0$ entails the measurement matrix $\matA$ satisfying certain isometry properties, which is elaborated upon in the following section. 
\fi


\ifdefined \SKIP
Next, in Proposition \ref{prop_supp_combi_cnt}, we show that 
$|\setS_{K}|$, the second term in the error exponent in 
\eqref{perr_abstract_bound}, depends only on $n$ and $K$, 
and it grows at most polynomially with $n$.
\begin{proposition} \label{prop_supp_combi_cnt}
Let $\setS_{K} = \lc \setS \subseteq [n], |\setS| \le K \rc$, for 
$K, n \in \mathbb{N}, K \le n$. Then, $|\setS_{K}| \le \displaystyle \frac{3}{2}(en)^{K}$.
\end{proposition}
\begin{proof}
The proof is straightforward. It uses Striling's approximation to bound the partial binomial sum. 
\end{proof}
\fi

\section{Exact Support Recovery \textcolor{black}{Using Sparse Bayesian Learning} \\ - Sufficient Conditions } 
\label{sec:sufficient_conditions}
Define $\mathcal{M}^{k}_{m,n}(\alpha, \beta)$ to be the set of all $m \times n$ 
sized real valued \textcolor{black}{sensing} matrices $\matA$ satisfying the following two properties:
\begin{enumerate}
\item $ \| \veca_{i} \|_{2}^{2} \in [1 - \alpha, 1 + \alpha], \forall i \in 
[n]$, 
where $\veca_{i}$ denotes the $i^{\text{th}}$ column of $\matA$
\item $\|(\matA\odot \matA) \vecv \|_{2} \ge \beta \| \vecv \|_{2}^{2}$ for 
all 
$k$ or less sparse vectors $\vecv \in \Real^{n}$.
\end{enumerate}
Equipped with the newly defined set $\mathcal{M}^{k}_{m,n}(\alpha, \beta)$ and the explicit bounds for $\eta$ and $\kappa_{\text{cov}}$ in Propositions~\ref{prop_eta_lb} and~\ref{prop_kcov_ub}, respectively,  
we now state the sufficient conditions for vanishing support error probability in cM-SBL. 

\begin{theorem} \label{thm_suff_conditions_form1}
Suppose $\matX$ has row support $\setS^{*}$, $|\setS^{*}| \le~K$, and 
satisfies assumption $A1$. 
Let $\hat{\vgamma}$ denote a solution of the cM-SBL optimization in \eqref{defn_cmsbl}. Then, for any $\delta \in (0,\frac{1}{2})$, $\text{supp}\lb \hat{\vgamma} \rb = \setS^{*}$ with probability exceeding $1-2\delta$, provided the following two conditions are satisfied.
\begin{description}
 \item[C1.] The sensing matrix $\matA \in \mathcal{M}^{K_{o}}_{m,n}(\alpha, 
 \beta)$, 
 where $K_{o}= K \!\lb \! 1 \!+\! 4 \! \lb \frac{\vgammamax}{\vgammamin} \rb  \!
 \lb \frac{1+ \alpha}{1 - \alpha}\rb^{2} \rb$ for some $\alpha \in (0,1)$ and $\beta > 0$.
 \item[C2.] The number of MMVs, $L$, satisfies
 \begin{align}
 \hspace{-0.25cm}
 L  &\ge  \frac{8}{\eta} \max \lc \log\lb \frac{6enK}{\delta} \rb,  \kappa_{\text{cov}} \rc,
 \nonumber
\label{cmsblK_sufficient_condition}
\end{align} 
where $\eta$ and $\kappa_{\text{cov}}$ depend on $\alpha$ and $\beta$ 
as described by the Propositions~\ref{prop_eta_lb} and \ref{prop_kcov_ub}, respectively.
\end{description}
\end{theorem}
\begin{proof}
Under condition \textbf{C1}, $\eta$ as defined in \eqref{eta_char} is rendered strictly positive due 
to Proposition~\ref{prop_eta_lb}. Further, condition \textbf{C2} ensures that the abstract 
MMV bound in Theorem~\ref{thm_abstract_suff_condition} is satisfied. Therefore, it follows directly from Theorem \ref{thm_abstract_suff_condition} that $\PP(\mathcal{E}_{\setS^*}) \le 2 \delta$.
\end{proof}
The following corollary of Theorem~\ref{thm_suff_conditions_form1} states 
an additional condition besides \textbf{C1} and \textbf{C2} that 
guarantees support consistency of any solution of the M-SBL optimization.
\begin{corollary}[Exact support recovery in M-SBL] 
\label{corr_transfer_results_to_msbl}
For the same setting as Theorem~\ref{thm_suff_conditions_form1}, 
let $\hat{\vgamma}$ denote any global solution of the M-SBL optimization in 
\eqref{ml_estimate_of_vgamma}. 
If $\hat{\vgamma} \in \Theta_{n}$, and the conditions \textbf{C1} and \textbf{C2} 
hold for any $\delta \in (0,\frac{1}{2})$, then 
$\text{supp}(\hat{\vgamma}) = \setS^{*}$ with probability exceeding $1 - 2\delta$. 
\end{corollary}
\begin{proof}	
Since $\hat{\vgamma}$ belongs to $\Theta_{n}$ and maximizes 
the M-SBL objective $\log{p(\matY; \vgamma)}$, it follows that $\hat{\vgamma}$ 
is also a solution to the constrained cM-SBL optimization in 
\eqref{defn_cmsbl}. 
Therefore, the statement of Corollary \ref{corr_transfer_results_to_msbl} follows from 
Theorem~\ref{thm_suff_conditions_form1}.
\end{proof}
By Corollary~\ref{corr_transfer_results_to_msbl}, retrospective to the nonzero coefficients of the M-SBL solution 
$\hat{\vgamma}$ 
lying inside $[\vgammamin, \vgammamax]$, the support error probability vanishes 
under conditions \textbf{C1} and \textbf{C2}. \textcolor{black}{It is to be noted that any probablistic guarantees of support consistency of the M-SBL solution, $\hat{\vgamma}$, for the case where $\hat{\vgamma} \notin \Theta_{n}$ are yet to be established. 
We conjecture that any global solution of the M-SBL optimization belongs to the set $\Theta_{n}$ with overwhelming probability, where $\vgammamin$ and $\vgammamax$ are constants independent of the MMV problem dimensions. However, showing such a result remains an open problem for future research.   
}

Compared to the existing M-SBL optimization-based support recovery guarantees in \cite{NehoraiGTang10} and \cite{Koochakzadeh18CorrAwareMMV}, 
which assume that $\hat{\vgamma}$ and $\vgamma^*$ are $K$-sparse binary valued vectors, the sparsistency guarantees derived here are applicable with wider 
scope. Our guarantees are valid for any 
real-valued $K$-sparse $\vgamma^*$, and 
for any M-SBL solution $\hat{\vgamma}$ provided its nonzero coefficients 
are bounded between $\vgammamin$ 
and $\vgammamax$. 

\ifdefined \SKIPTEMP
In the above, we have analyzed the support error probability for the case where the 
true row support of $\matX$ is fixed to $\setS^*$. Noting that $\mathcal{R}(\matX)$ is equally likely to assume any one of the $\binom{n}{1} + \binom{n}{2} + \dots + \binom{n}{K}$ possible supports in the collection $\setS_{K}$, the average support error probability can be evaluated as
\begin{equation} \label{defn_avg_prob_err}
 P_{\text{avg}}^{\text{err}} = 
 \sum_{\setS^* \in \setS_{K}} \PP(\mathcal{R}(\matX)  = \setS^*)  
 \PP \lb  \mathcal{E}_{\setS^*} \rb, 
\end{equation}
where $\mathcal{E}_{\setS^*}$ is as defined in \eqref{bad_mmv_set}. 
For example, if all supports in $\setS_{K}$ are equiprobable, 
then $P_{\text{avg}}^{\text{err}} \le 2\delta$ 
under \textbf{C1} and \textbf{C2}.
\fi

\section{Discussion}\label{sec:discussion}
We now consider a few special cases of the MMV problem and reinterpret the sufficient conditions specified in Theorem~\ref{thm_suff_conditions_form1} which guarantee vanishing support error probability. 

\ifdefined \SKIPTEMP
\subsection{Column Normalized Measurement Matrices} 
According to \textbf{C2} in Theorem \ref{thm_suff_conditions_form1}, the number of 
sufficient MMVs grows quadratically with 
$a_{*}$, the largest absolute value of entries in $\matA$. 
If $\matA$ has unit norm columns, 
then $a_{*}$ scales inversely with the number of rows in $\matA$ 
as $\mathcal{O}\lb\frac{1}{\sqrt{m}}\rb$. 
For example, if $\matA$ is a random matrix 
with i.i.d. $\mathcal{N}(0, \frac{1}{\sqrt{m}})$ entries, 
then $|a_*| \le \mathcal{O}\lb \sqrt{\frac{\log n}{m}} \rb$ with high probability. 
Likewise, if $\matA$ is a partial DFT sensing matrix, then $|a_*| = \frac{1}{\sqrt{m}}$. 

Often, column normalization is essential for accurate modelling of 
the norm-preserving nature of the application specific measurement modality.
In light of this, if $\matA$ has normalized columns, then one must 
account for $\mathcal{O}(\frac{1}{\sqrt{m}})$ scaling of $a_*$ while 
interpreting the sufficient number of MMVs in $\textbf{C2}$. 
\fi

\ifdefined \SKIPTEMP
\subsection{Continuous versus Binary Hyperparameters}
In \cite{NehoraiGTang10, Koochakzadeh18CorrAwareMMV}, the support recovery is 
formulated as a multiple hypothesis testing problem by assuming that each column of $\matX$ is i.i.d. $\mathcal{N}(0, \diag{(\mathbf{1}_{\setS^{*}})})$, where 
$\setS^{*}$ is the $k$-sparse support of $\matX$. For this choice of signal prior, 
finding the true support via type-II likelihood maximization 
as in~\eqref{ml_estimate_of_vgamma} is no longer a continuous variable optimization 
but rather a combinatorial search over all 
$k$-sparse vertices of the hypercube $\lc 0, 1\rc^{n}$, as described below.
\begin{equation}
 \hat{\vgamma} = \underset{\vgamma \in \lc 0, 1\rc^{n}, \lvv \vgamma \rvv_{0} = K}{\text{arg max }} \log{p \lb \matY; \vgamma \rb} .
\end{equation}
In \cite{NehoraiGTang10}, it has been shown that for a $2K$-RIP compliant $\matA$, 
and binary hyperparameters $\vgamma$, $\mathcal{O}\lb \frac{\log n}{\log \log n}\rb$ MMVs suffice for the support error probability to vanish. 
Binary valued hyperparameters can be accommodated as a special case 
by setting $\vgammamin = \vgammamax =~1$ in our source signal model. 
For $\vgammamin = \vgammamax$, according to Proposition~\ref{prop_cov_num_enet}, 
the $\epsilon$-net $\Theta^{\epsilon}(\setS)$ collapses to a single point for all 
$\setS \in \setS_{K}$, which in turn implies that $\kappa_{\text{cov}} = 1$.
Thus, for the binary hyperparameter case, by setting $\kappa_{\text{cov}} = 1$ in \textbf{C2}, 
it can be concluded that $L \ge \max \lb \Omega (\log{n}), \Omega\lb \frac{m \log n}{K} \rb \rb $ suffices for vanishing support error probability. 
\fi

\subsection{The case of binary hyperparameters and known support size}
In \cite{NehoraiGTang10} and \cite{Koochakzadeh18CorrAwareMMV}, support recovery is treated as a multiple hypothesis testing problem by assuming that each column of~$\matX$ is i.i.d. $\mathcal{N}(0, \diag{(\mathbf{1}_{\setS^{*}})})$, where $\setS^{*}$ is the true $K$-sparse row-support of $\matX$. In 
this case, finding the true support via type-II likelihood maximization as in~\eqref{ml_estimate_of_vgamma} can be reformulated as a combinatorial search over all 
$k$-sparse vertices of the hypercube $\lc 0, 1\rc^{n}$, as described below.
\begin{equation}
\hat{\vgamma} = \underset{\vgamma \in \lc 0, 1\rc^{n}, \| \vgamma \|_{0} \le K}{\text{arg max }} \log{p \lb \matY; \vgamma \rb} .
\end{equation}
The binary valued hyperparameters can be accommodated as a special case 
of our source model by simply setting $\vgammamin = \vgammamax =~1$. 
For $\vgammamin = \vgammamax$, according to Proposition~\ref{prop_cov_num_enet}, 
the $\epsilon$-net $\Theta^{\epsilon}(\setS)$ collapses to a single point for all 
$\setS \in \setS_{n}$, which ultimately amounts to $\kappa_{\text{cov}} = 0$.
In~\cite{NehoraiGTang10, Koochakzadeh18CorrAwareMMV}, the correct support has 
to be identified from $\binom{n}{K}$  candidate support hypothesis. Under this 
restrained setting, the lower bound for $\eta$ in \eqref{eta_lb} simplifies to 
\begin{equation}
\eta \ge 
\frac{\beta}{4(\sigma^2 + 1)^2 \delta_{2K}^2},
\label{eta_for_binary_vgamma}
\end{equation}
where $\beta$ denotes the smallest restricted singular value of \mbox{$\matA \odot \matA$} of order $2K$.
By setting $\kappa_{\text{cov}} = 0$  and using \eqref{eta_for_binary_vgamma} 
in Theorem~\ref{thm_suff_conditions_form1}, we can conclude that  
$\PP(\mathcal{E}_{\setS^*})$ is at most $2\epsilon$ if   
\begin{equation}
L \ge  \frac{32(\sigma^2 + 1)^2 \delta_{2K}^2}{\beta} \log \lb \frac{6enK}{\epsilon} \rb.
\label{eqn_mmv_sufficient_binary_gamma}
\end{equation}
For the special case where the sensing matrix $\matA$ comprises i.i.d. 
$\mathcal{N}\lb 0, \frac{1}{m}\rb$ entries and has $m = \Theta(K \log 
n)$ rows, we have $\delta_{2K} = O(1)$ from 
Corollary~\ref{corr_Gaussian_submatrix_spectral_bound}, and $\beta = \Omega(1)$ 
from \cite[Theorem $3$]{Sak19ErratumKRRIC}. Therefore, it follows from 
\eqref{eqn_mmv_sufficient_binary_gamma} and 
Corollary~\ref{corr_transfer_results_to_msbl} that $L \ge \Omega\lb \log n \rb$ 
suffices to ensure that any $K$ or less sparse binary vector in $\Real^{n}$ that maximizes 
the M-SBL objective also recovers the true row-support of~$\matX$ with overwhelming 
probability. In this regard, our MMV bound matches with the one derived 
in~\cite{NehoraiGTang10}, $L \gg \frac{\log n}{\log \log n}$, up to an additional $\log \log n$ 
factor. 

A more interesting case arises when $m = \Theta(\sqrt{K} \log n)$. From \eqref{eqn_random_A_results_intr0a} and \eqref{eqn_random_A_deltak_ub2} in Appendix~\ref{app:proof_corr_gaussian_meas_map_eta_char}, we observe that $\beta = \Omega(1)$ and $\delta_{2K} = O(\sqrt{K})$. Substituting this in \eqref{eqn_mmv_sufficient_binary_gamma}, we conclude that $L = \Omega(K \log n)$ suffices to ensure that an M-SBL solution from the restrained set of $K$-sparse binary vectors recovers the true support with overwhelming probability.

\subsection{The case of continuous-valued hyperparameters and unknown support size}
The sparsitency of a continuous-valued solution of the M-SBL optimization has been investigated in \cite{alex2019nonbayesian} under the assumption that the nonzero amplitudes of $\vgamma^*$ are known a priori and their guarantees apply to the case where the search for $\vgamma$ is restricted to $K$ or less sparse vectors in 
$\Theta_{K}$. Our analysis dispenses with both of these restrictive assumptions. 

When the sensing matrix $\matA$ consists of i.i.d. $\mathcal{N}\lb 0, \frac{1}{m} \rb$ entries, by substituting $\eta$ and~$\kappa_{\text{cov}}$ in Theorem~\ref{thm_suff_conditions_form1} with their 
respective lower and upper bounds from Corollary~\ref{corr_gaussian_meas_map_eta_char}, we see that the M-SBL solution in~$\Theta_{n}$ recovers the true support with vanishing support error probability provided $m = \Theta(K \log n)$ and $L = \Omega\lb \frac{n}{K} \log n + n \log K + n \log \log n \rb$. On the other hand, when $m = \Theta(\sqrt{K} \log n)$, $L = \Omega \lb \frac{n}{\sqrt{K}} \log n + n \sqrt{K} \log K + n \sqrt{K} \log \log n \rb$ is sufficient.

By restricting the hyperparameter search in M-SBL to $\Theta_{K}$, in \cite{alex2019nonbayesian}, it is shown that the $L = \Omega\lb K \log^{2}\lb\frac{n}{K}\rb \log(nK) \rb$ is sufficient for exact support recovery when $m = \Omega(\sqrt{K} \log n)$. Comparing with our MMV bound, the same value of $m$ is sufficient, but we pay an extra penalty of a roughly $\frac{n}{\sqrt{K}}$ factor in the MMV complexity, in order to circumvent the 
restrictive assumptions on the hyperparameter search domain in \cite{alex2019nonbayesian}.


\subsection{The case when $\matA \odot \matA$ is full column rank} 
The works in \cite{PPal15PushingTheLimits} and 
\cite{BhaskaraMoitra14SmoothedAna} have discussed random as well as 
deterministic constructions of the sensing matrix 
$\matA$ for which the $m^2 \times n$ sized self Khatri-Rao product 
$\matA~\odot~\matA$ is full column rank provided $n \le \frac{m^2 + 
m}{2}$. In this case, when the columns of $\matA$ are approximately normalized, $\matA \odot \matA$ satisfies the 
$K^{\text{th}}$ order NN-RNSP (Definition~\ref{defn_rnsp}) by default for any $1\le K \le n$, which 
in turn implies that $\eta$ in \eqref{eta_char} is always strictly positive. 
Hence, by Theorem~\ref{thm_abstract_suff_condition}, it follows that 
the support error probability decays exponentially with the number of MMVs even for $K = 
\Theta(m^2)$.

\subsection{The case when $\matA \odot \matA$ is rank deficient} 
For $n > \frac{m^2 + m}{2}$, the $m^2 \times n$ sized $\matA \odot \matA$ is 
rank-deficient. In \cite{Koochakzadeh18CorrAwareMMV} and \cite{PPal14SBLParamIdentifiability}, it is argued 
that this leads to parameter identifiability issues in M-SBL based support 
reconstruction. Specifically, the M-SBL objective can attain the same value 
for multiple distinct $\vgamma$, thereby fostering multiple 
global maxima with potentially different supports. 
However, both \cite{Koochakzadeh18CorrAwareMMV} and \cite{PPal14SBLParamIdentifiability} do not take the non-negativity of the hyperparameters in~$\vgamma$ into consideration. In Proposition~\ref{prop_eta_lb}, we have shown that as long as 
$\matA \odot \matA$ satisfies NN-RNSP of order $K$, $\eta$ is always strictly positive and 
consequently $K$-sized supports can be recovered exactly with arbitrarily 
high probability using finitely many MMVs. Interestingly, the NN-RNSP 
condition can hold even when $\matA \odot \matA$ is column 
rank deficient, which allows cM-SBL to recover any $K = O(m^2)$-sized supports 
exactly even when $ n > \frac{m^2 + m}{2}$. 


\ifdefined \SKIPTEMP
This result is valid as long as $\matA \odot \matA$ satisfies RIP of order $2K$. From the previous chapter, we know that for random $\matA$ with i.i.d. subgaussian entries, $\matA \odot \matA$ satisfies the $2K$ order RIP w.h.p. pro vided $m = \mathcal{O}\lb \sqrt{K}\log^{1.5} n \rb$. Hence, as per our result, if 
$m = \mathcal{O} \lb \sqrt{K} \log^{1.5} n \rb$ and $L = \mathcal{O}\lb \frac{K}{\sqrt{\log n}}\rb$, successful support recovery is possible w.h.p. In contrast, the sufficient conditions derived in \cite{NehoraiGTang10} require that $m =\Omega \lb K \log  \frac{n}{K}  \rb$, which is a significantly stronger 
condition on $m$ than ours.  
\fi


\subsection{The case of noiseless measurements}
If $K < \text{spark}(\matA) -1$, it can be shown that as the noise variance $\sigma^{2} \to 0$, 
the error exponent $\DD_{\frac{1}{2}}\lb p_{\vgamma}, p_{\vgamma^*}\rb$ 
in~\eqref{large_dev_event} grows unbounded in the event of a support mismatch, $\text{supp}(\vgamma) \neq \text{supp}(\vgamma^*)$, culiminating in vanishing support error probability even when $L = 1$.
This is formally proved below.

The $1/2$-R\'{e}nyi divergence between two multivariate Gaussian 
densities $p_{\vgamma_{i}}(\vecy) \sim \mathcal{N}(0, \matSigma_{\vgamma_{i}})$, 
$i = 1, 2$ is given by 
\begin{eqnarray}
\DD_{\frac{1}{2}} \lb p_{\vgamma_{1}}, p_{\vgamma_{2}} \rb 
&=&
\log{\lv \frac{\matSigma_{\vgamma_{1}} + \matSigma_{\vgamma_{2}}}{2} \rv} 
- \frac{1}{2} \log{\lv \matSigma_{\vgamma_{1}} \matSigma_{\vgamma_{2}} \rv} 
\nonumber \\
& = & 
\log{\lv \frac{ \matH^{1/2} + \matH^{-1/2} }{2} \rv}. 
\end{eqnarray}
where $\matH \triangleq \matSigma_{\vgamma_{1}}^{1/2} 
\matSigma_{\vgamma_{2}}^{-1} \matSigma_{\vgamma_{1}}^{1/2}$ is referred to as 
the \emph{discrimination matrix}. 
Since $\matH$ is a normal matrix, it is unitarily diagonalizable. Let $\matH = \matU \matLambda \matU^{T}$, where 
$\matLambda = \diag{\lb \lambda_{1}, \ldots, \lambda_{m} \rb}$ with $\lambda_{i}$'s being the strictly positive eigenvalues of $\matH$ for any $\sigma^2 > 0$, 
and $\matU$ being a unitary matrix with the eigenvectors of~$\matH$ as its columns. The $1/2$-R\'{e}nyi divergence can be expressed 
in terms of $\lambda_{i}$ as 
\begin{align}
\DD_{\frac{1}{2}} \lb p_{\vgamma_{1}}, p_{\vgamma_{2}} \rb &=
\sum_{i = 1}^{m} \log{\lb \lb \lambda_{i}^{1/2} + \lambda_{i}^{-1/2} \rb/2 \rb} 
\nonumber \\
& \hspace{0cm} \geq
\log{\lb \frac{1}{2} \lb \lb \lambda_{\text{max}}(\matH) \rb^{1/2} + \lb \lambda_{\text{max}}(\matH) \rb^{-1/2} \rb \rb }.
\label{renyi_div_lb_form2}
\end{align}
The above inequality is obtained by dropping all positive terms in the summation except the one term which corresponds to 
$\lambda_{\text{max}}(\matH)$, the maximum eigenvalue of $\matH$. Proposition~\ref{renyi_div_lb_gaussian_case_form2} below 
relates $\lambda_{\text{max}}(\matH)$ to the noise variance $\sigma^{2}$.
\begin{proposition} \label{renyi_div_lb_gaussian_case_form2} 
If $K < \text{spark($\matA$)}-1$, then for any $\vgamma_{1}, \vgamma_{2} \in 
\Real^{n}_{+}$ such that $\text{supp}(\vgamma_{1}) \neq 
\text{supp}(\vgamma_{2})$ and $|\text{supp}(\vgamma_{2})| \le K$, the maximum 
eigenvalue of $\matH \triangleq \matSigma_{\vgamma_{1}}^{1/2} 
\matSigma_{\vgamma_{2}}^{-1} \matSigma_{\vgamma_{1}}^{1/2}$ satisfies
\begin{equation}
\lambda_{\text{max}} \lb \matH \rb \ge \frac{c_{1}}{\sigma^{2}}
\nonumber
\end{equation}
for some constant $c_{1} > 0$ independent of $\sigma^{2}$.
\end{proposition}
\begin{proof}
	See Appendix \ref{app:proof_renyi_div_lb_gaussian_case_form2}.
	\nonumber
\end{proof}
According to Proposition~\ref{renyi_div_lb_gaussian_case_form2}, 
in the limit $\sigma^{2} \to 0$, $\lambda_{\text{max}}(\matH) \to \infty$, and 
consequently, $\DD_{\frac{1}{2}} \!\lb p_{\vgamma_{1}}\!, p_{\vgamma_{2}} \!\rb\!$ grows unbounded 
(due to~\eqref{renyi_div_lb_form2})
whenever $\text{supp}(\vgamma_{1}) \neq \text{supp}(\vgamma_{2})$ and 
$K < \text{spark}(\matA)-1$. 
Based on this observation, 
we now state Theorem~\ref{thm_suff_conditions_form2} which lays forward the sufficient conditions for exact support recovery in the noiseless case.\\
\begin{theorem} \label{thm_suff_conditions_form2} 
Consider the noiseless MMV problem, with observations $\matY = \matA \matX$ corresponding to an unknown $\matX$ satisfying assumption \textbf{A1}.
Suppose $\setS^{*}$ is the true nonzero row support of $\matX$ with 
$|\setS^{*}| \le K$. Further, let $\hat{\vgamma}$ be a solution of the MSBL optimization in \eqref{ml_estimate_of_vgamma}, then $\text{supp}(\hat{\vgamma}) = \setS^{*}$ 
almost surely, provided that $K < \text{spark}(\matA)-1$. This result holds even in the SMV case, i.e., when $L=1$.
\end{theorem}
\begin{proof} 
Under assumption~\textbf{A1}, there exists a $\vgamma^{*} \in \Theta_{K}$ 
such that every column in $\matX$ is i.i.d. $\mathcal{N}(0, \diag{(\vgamma^{*})})$, and 
$\text{supp}(\vgamma^{*}) =~\setS^{*}$. 
Since $\hat{\vgamma}$ globally maximizes the MSBL objective 
$\LL(\matY; \vgamma)$, it follows that 
$\LL(\matY; \hat{\vgamma}) \ge \LL(\matY; \vgamma^{*})$ 
if $\hat{\vgamma} \neq \vgamma^{*}$, i.e., the following chain of implications holds.
\begin{eqnarray}
\lc \text{supp}(\hat{\vgamma}) \neq \setS^{*}\rc &=& 
\lc \text{supp}(\hat{\vgamma}) \neq \text{supp}(\vgamma^{*}) \rc
\nonumber \\
&\subseteq & \lc \hat{\vgamma} \neq \vgamma^{*} \rc 
\nonumber \\
&\subseteq& \lc \LL \lb \matY; \hat{\vgamma} \rb \ge \LL \lb \matY; \vgamma^{*} \rb \rc.
\nonumber
\end{eqnarray}
By applying Corollary~\ref{coro_model_mismatch_ldp}, this further implies that 
\begin{eqnarray} 
\PP\lb \text{supp}(\hat{\vgamma}) \neq \setS^{*} \rb 
&\le &
\PP \lb \LL \lb \matY; \hat{\vgamma} \rb \ge \LL \lb \matY; \vgamma^{*} \rb  \rb 
\nonumber \\
&\le&
\exp{\lb - \frac{L \DD_{\frac{1}{2}} \lb p_{\hat{\vgamma}}, p_{\vgamma^{*}} \rb}{4} \rb}. 
\nonumber
\end{eqnarray}
By using the lower bound in \eqref{renyi_div_lb_form2} for 
$\DD_{\frac{1}{2}} \lb p_{\hat{\vgamma}}, p_{\vgamma^{*}} \rb  $, we have
\begin{equation}
\PP\lb \text{supp}(\hat{\vgamma}) \neq \setS^{*} \rb \le
\ls \frac{1}{2} \lb \sqrt{\lambda_{\max}(\matH)} + \frac{1}{\sqrt{\lambda_{\max}(\matH)}} \rb \rs^{-\frac{L}{4}},
\label{noiseless_perr_ub}
\end{equation}
where $\matH = \matSigma_{\hat{\vgamma}}^{1/2} \matSigma_{\vgamma^{*}}^{-1} \matSigma_{\hat{\vgamma}}^{1/2} $.
Since $\vgamma^{*}$ is at most $K$-sparse, 
as long as $K < \text{spark}(\matA) -1$, by Proposition~\ref{renyi_div_lb_gaussian_case_form2}, 
$\sigma^{2} \to 0$ results in $\lambda_{\text{max}}\lb \matH \rb \to \infty$ 
which in turn drives the RHS in~\eqref{noiseless_perr_ub} to zero for $L \ge 1$. 
\end{proof}

From Theorem~\ref{thm_suff_conditions_form2}, we conclude that, in the noiseless scenario ($\sigma^2 \to 0 $) and for $\matX$ satisfying assumption \textbf{A1}, MSBL requires only a single measurement vector ($L=1$) to perfectly recover any $K < \text{spark}(\matA) -1$ sized support. 
If the sensing matrix $\matA$ has full spark, i.e., 
$\text{spark}(\matA) = m+1$, MSBL can recover $m-1$ or lesser 
sparse supports exactly from $m$ noiseless measurements of a single sparse 
vector. It is noteworthy that $\matA$ satisfies the full spark condition under very mild assumptions, e.g., the entries of $\matA$ are drawn independently from a continuous probability distribution. This result is an improvement over the sufficient conditions for exact support recovery by MSBL in \cite[Theorem 1]{Wipf_07_msbl}. Unlike in \cite{Wipf_07_msbl}, we do not require the nonzero rows of $\matX$ to be orthogonal. Also, our result improves over the $k \le \frac{m}{2}$ condition shown in~\cite{NehoraiGTang10}. 

\ifdefined \SKIPTEMP
\subsection{Impact of Measurement Noise on Sufficient MMVs}
For $\sigma^{2} > 0$, the error exponent term $\DD_{\setS}^*$ in 
\eqref{eqn_minD} is never unbounded. As a consequence, unlike in the noiseless 
case, a single MMV is no longer sufficient, 
and multiple MMVs are needed to drive the support error probability to zero.

A close inspection of the MMV bound in Theorem~\ref{thm_suff_conditions_form1} 
reveals that the noise variance influences the error probability in a twofold manner: 
(i) through $\eta$, and (ii) through the size of the $\epsilon$-net or $\kappa_{\text{cov}}$. 
As $\sigma^{2}$ increases, $\eta$ decreases polynomially (see~\eqref{eta_lb}) while $\kappa_{\text{cov}}$ increases at most logarithmically (Proposition~\ref{prop_kcov_ub}). The overall effect 
is captured by condition~\textbf{C2} in Theorem~\ref{thm_suff_conditions_form1} that 
if $\sigma^2$ is very high relative to $\vgammamax \delta_K$, the number of MMVs that are sufficient to guarantee the desired probability of error grows quadratically with $\sigma^2$.
As the noise variance approaches zero, the MMV bound in~\textbf{C2} loosens and is not informative. 
\fi

\section{Conclusions}\label{sec:conclusion}

We analyzed the sample complexity of error free recovery of the common 
nonzero support of multiple joint-sparse Gaussian sources from their 
compressive measurements. We established the finite MMV, high-probability 
consistency of the nonzero support inferred from a constrained \mbox{type-II 
ML} estimate of the variance hyperparameters belonging to a correlation-aware 
Gaussian source prior. The nonzero coefficients of the type-II ML estimate are constrained to lie in 
a known interval $[\vgammamin, \vgammamax]$. Our support consistency guarantee also applies to any global solution of the M-SBL optimization when the nonzero coefficients of the solutions satisfy the same interval constraint. 

We also showed that a single noiseless MMV suffices for perfect recovery of 
any $K$-sparse support, provided that $K < \text{spark}(\matA)-1$. In case 
of noisy MMVs, we showed that any $K$-sized support can be recovered exactly with high probability using finitely many MMVs, provided $\matA \odot \matA$ admits a positive minimum restricted singular value of order $\Theta(K)$. We also presented an interesting interpretation of M-SBL's marginalized log-likelihood cost as a Bregman matrix divergence, which highlights that the M-SBL algorithm is, in principle, a covariance 
matching algorithm. 
There still remain the following open questions regarding M-SBL-based support recovery: 
\begin{enumerate}[label=(\roman*)]
 \item What are necessary conditions for exact support recovery in terms of the number of required MMVs?
 \item Is there a criterion under which all stationary points of the M-SBL objective also yield the correct support estimate?
 \item How is the support recovery performance impacted by inter and intra vector 
 correlations in the signals?
\end{enumerate}
Answering these questions would be interesting directions for future work.

\appendix

\subsection{Proof of Proposition~\ref{prop_renyi_div_lb_analytical}} \label{app:proof_thm_renyi_div_lb_analytical}
By using the property: $\log|\matX \matY| = \log|\matX| + \log|\matY|$ for any positive definite matrices $\matX$ and $\matY$, From \eqref{renyi_div_multivar_gauss}, we have
\begin{eqnarray}
\DD_{\frac{1}{2}}(p_{1}, p_{2}) &=& \log \lv \frac{\matSigma_{1} + \matSigma_{2}}{2} \rv 
- \frac{1}{2} \log \lv \matSigma_{1} \rv - \frac{1}{2} \log \lv \matSigma_{2} \rv
\nonumber \\
& \hspace{0cm} =& \log \lv \frac{\matI_{m} + \matSigma_{1}^{-1/2} \matSigma_{2} \matSigma_{1}^{-1/2} }{2} \rv 
+ \frac{1}{2} \log \lv \matSigma_{1} \rv - \frac{1}{2} \log \lv \matSigma_{2} \rv
\nonumber \\
& \hspace{0cm} =& \log \lv \frac{\matI_{m} + \matSigma_{1}^{-1/2} \matSigma_{2} \matSigma_{1}^{-1/2} }{2} \rv 
+ \frac{1}{2} \log \lv \matSigma_{1}^{1/2} \matSigma_{2}^{-1} \matSigma_{1}^{1/2} \rv
\nonumber \\
& \hspace{0cm} =& \log \lv \frac{\matI_{m} + \matH^{-1}}{2} \rv 
+ \frac{1}{2} \log \lv \matH \rv 
= \log \lv \frac{\matH^{1/2} + \matH^{-1/2}}{2} \rv,
\label{eqn_ren_div_lb_proof1}
\end{eqnarray}
where $\matH \triangleq \matSigma_{1}^{1/2} \matSigma_{2}^{-1} \matSigma_{1}^{1/2}$. 
Denote the eigenvalue decomposition of $\matH$ by $\matU \Lambda \matU^{T}$, with the real and positive diagonal entries of $\Lambda = \diag \lb \lambda_{1}, \lambda_{2}, \ldots, \lambda_{m} \rb$ denoting the eigenvalues of $\matH$, and 
the orthonormal columns of $\matU$ denoting the eigenvectors of $\matH$. 
\textcolor{black}{By the non-negativity of the Bregman Log-Det divergence between positive definite matrices \cite{ArindamBanerjee05ClusteringWithBregman}, for any $\matX \in S^{m}_{++}$, it follows that} $\log|\matX| \ge \text{tr}\lb \matI - \matX^{-1} \rb$. Therefore, using \eqref{eqn_ren_div_lb_proof1}, we have
\begin{eqnarray}
\DD_{\frac{1}{2}}\lb p_{1}, p_{2} \rb &\ge& 
\text{tr} \lb \matI_{m} - \lb \frac{\matH^{1/2} + \matH^{-1/2}}{2} \rb^{-1} \rb
\nonumber \\
&=&  
\sum_{i = 1}^{m} \lb 1 - \lb \frac{\lambda_{i}^{1/2} + \lambda_{i}^{-1/2} }{2}\rb^{-1} \rb
\ge
\sum_{i = 1}^{m} \frac{(\lambda_{i} -1)^2}{2(1 + \lambda_{i})^2} 
\nonumber \\
&=& 
\frac{1}{2} \text{tr} \lb (\matH - \matI_{m}) (\matI_{m} + \matH)^{-1} (\matH - \matI_{m}) (\matI_{m} + \matH)^{-1} \rb.
\label{eqn_ren_div_lb_proof2}
\end{eqnarray}

Plugging back $\matH = \matSigma_{1}^{1/2} \matSigma_{2}^{-1} \matSigma_{1}^{1/2}$ in 
\eqref{eqn_ren_div_lb_proof2}, we obtain the desired lower bound for $\DD_{\frac{1}{2}}(p_{1}, p_{2})$ as shown below.
\begin{eqnarray}
D_{\frac{1}{2}}\lb p_{1}, p_{2} \rb &\ge& 
\frac{1}{2}\text{tr} \lb (\matSigma_{1}^{1/2} \matSigma_{2}^{-1} \matSigma_{1}^{1/2} - \matI_{m}) (\matI_{m} + \matH)^{-1}
(\matSigma_{1}^{1/2} \matSigma_{2}^{-1} \matSigma_{1}^{1/2} - \matI_{m}) (\matI_{m} + \matH)^{-1}\rb
\nonumber \\
& =& \frac{1}{2}\text{tr} \lb \matSigma_{1}^{1/2} (\matSigma_{2}^{-1} - \matSigma_{1}^{-1}) 
\matSigma_{1}^{1/2} (\matI_{m} + \matH)^{-1} 
\matSigma_{1}^{1/2} (\matSigma_{2}^{-1} - \matSigma_{1}^{-1}) \matSigma_{1}^{1/2} (\matI_{m} + \matH)^{-1}\rb
\nonumber \\
&=& \frac{1}{2}\text{tr} \lb \matSigma_{1}^{-1} (\matSigma_{1} - \matSigma_{2}) \matSigma_{2}^{-1} \matSigma_{1}^{1/2} (\matI_{m} + \matH)^{-1} 
\matSigma_{1}^{1/2} \matSigma_{1}^{-1} (\matSigma_{1} - \matSigma_{2}) \matSigma_{2}^{-1} \matSigma_{1}^{1/2} (\matI_{m} + \matH)^{-1} \matSigma_{1}^{1/2} \rb
\nonumber \\
&=& \frac{1}{2}\text{tr} \ls \lb \matSigma_{1} - \matSigma_{2}\rb \lb \matSigma_{2}^{-1} \matSigma_{1}^{1/2} (\matI_{m} + \matH)^{-1} \matSigma_{1}^{-1/2}\rb
\lb \matSigma_{1} - \matSigma_{2}\rb \lb \matSigma_{2}^{-1} \matSigma_{1}^{1/2} (\matI_{m} + \matH)^{-1} \matSigma_{1}^{-1/2}\rb \rs
\nonumber \\
&=& \frac{1}{2}\text{tr} \ls \lb \matSigma_{1} - \matSigma_{2}\rb \lb \matSigma_{2}^{-1} (\matI_{m} + \matSigma_{1}\matSigma_{2}^{-1})^{-1} \rb 
\lb \matSigma_{1} - \matSigma_{2}\rb \lb \matSigma_{2}^{-1} (\matI_{m} + \matSigma_{1} \matSigma_{2}^{-1})^{-1} \rb \rs
\nonumber \\
&=& \frac{1}{2}\text{tr} \ls \lb \matSigma_{1} - \matSigma_{2}\rb \lb  
\matSigma_{2} + \matSigma_{1}\rb^{-1} \lb \matSigma_{1} - \matSigma_{2}\rb \lb \matSigma_{2} + \matSigma_{1} \rb^{-1} \rs.
\nonumber
\end{eqnarray}

\subsection{Proof of Corollary~\ref{corr_Gaussian_submatrix_spectral_bound}} 
\label{app:proof_corr_Gaussian_submatrix_spectral_bound}
For a fixed support $\setS, |\setS| \le k$, from Proposition~\ref{prop_gaussian_spectral_bound}, 
\begin{equation}
\PP \lb \lvvv \matA_{\setS} \rvvv_{2} \ge \sqrt{m} + \sqrt{k} + \sqrt{6k \log n} \rb  \le 
2e^{-3 k \log n}.
\nonumber
\end{equation}
By taking the union bound over $\binom{n}{1} + \binom{n}{2} + \ldots + \binom{n}{k} \le 
\lb \frac{3en}{2} \rb^{k}$ submatrices of $\matA$ containing $k$ or fewer columns, we get
\begin{eqnarray}
\PP \lb \bigcup_{\setS \subset [n]: |\setS| \le k} 
\lc \lvvv \matA_{\setS} \rvvv_{2} \ge \sqrt{m} + \sqrt{k} + \sqrt{6 k \log n}  \rc  \rb 
&\le&  \lb \frac{3ne}{2} \rb^{k} 2e^{-3 k \log n}
\nonumber \\
&\le &  2 e^{-3 k \log n  + k \log \lb 3ne/2 \rb } 
\le \frac{2}{n^{k}},
\nonumber
\end{eqnarray}
for $n > 5$.

\subsection{Proof of Proposition \ref{prop_cov_num_enet}} \label{app:proof_prop_cov_num_enet}
The following stepwise procedure shows how to construct a $\delta$-net of 
$\Theta(\setS)$ (with respect to the Euclidean distance metric) which is 
entirely contained in $\Theta(\setS)$.
\begin{enumerate}
	\item Consider an $\delta$-blow up of $\Theta(\setS)$, denoted by 
	\begin{equation}
	\Theta_{\uparrow\delta}(\setS) \triangleq \lc x: \exists x^{\prime} \in 
	\Theta(\setS) \text{ such that } \Vert x - x^{\prime}\Vert_{2} \le \delta \rc.
	\nonumber
	\end{equation}
	\item Let $\Theta^{\delta}_{\uparrow \delta}(\setS)$ be a $\delta$-net 
	of $\Theta_{\uparrow \delta}(\setS)$. Some points in 
	$\Theta^{\delta}_{\uparrow \delta}(\setS)$ may lie outside $\Theta(\setS)$.

	\item Let $\setP$ denote the set containing the projections of all points 
	in $\Theta^{\delta}_{\uparrow \delta}(\setS) \cap \Theta(\setS)^{c}$ 
	onto the set $\Theta(\setS)$. By construction, $\setP \subset \Theta(\setS)$, 
	and $|\setP| \le |\Theta^{\delta}_{\uparrow \delta}(\setS) \cap \Theta(\setS)^{c}|$.

	\item Then, $\Theta^{\delta}(\setS) \triangleq 
	\lb \Theta^{\delta}_{\uparrow \delta}(\setS) \cap \Theta(\setS)\rb
	\cup \setP$ is a valid $\delta$-net of $\Theta(\setS)$ which is entirely 
	contained in $\Theta(\setS)$.	
\end{enumerate}

To prove the validity of the above $\delta$-net construction, we need to show that 
for any $\vgamma \in \Theta(\setS)$, there exists an element $\veca$ in $\Theta^{\delta}(\setS)$ such that 
$\Vert \vgamma - \veca \Vert_{2}\le \delta$. Let $\vgamma$ be an arbitrary element 
in $\Theta(\setS)$. Then, $\vgamma$ also belongs to the larger set $\Theta_{\uparrow \delta}(\setS)$, and consequently, 
there exists $\vgamma^{\prime} \in \Theta_{\uparrow \delta}^{\delta}(\setS)$ such that $\Vert \vgamma - \vgamma^{\prime}\Vert_{2} \le \delta$. Now, there are two cases. 
(i) $\vgamma^{\prime} \in \Theta(\setS)$, and (ii) $\vgamma^{\prime} \notin \Theta(\setS)$. 

In case (i), $\vgamma^{\prime} \in \lb \Theta_{\uparrow \delta}^{\delta}(\setS) 
\cap \Theta(\setS)\rb$, and hence also belongs to $\Theta^{\delta}(\setS)$. Further, 
$\Vert \vgamma - \vgamma^{\prime}\Vert_{2} \le \delta$. Hence $\veca = \vgamma^{\prime}$ 
will work.

In case (ii), $\vgamma^{\prime} \in \Theta_{\uparrow \delta}^{\delta}(\setS) \cap 
\Theta(\setS)^{c}$. 
Let $\vgamma^{\prime \prime}$ be the projection of $\vgamma^{\prime}$ onto
$\Theta(\setS)$, then $\vgamma^{\prime \prime}$ must belong to $\setP$, and hence 
must also belong to $\Theta^{\delta}(\setS)$. 
Note that since $\vgamma^{\prime\prime}$ is the projection of $\vgamma^{\prime}$ 
onto the \emph{convex} set $\Theta(\setS)$, for any $\vgamma \in \Theta(\setS)$, we have 
$\langle \vgamma - \vgamma^{\prime\prime}, \vgamma^{\prime} - \vgamma^{\prime\prime}\rangle 
\le 0$. Further, we have 
\begin{eqnarray}
\delta &\ge& 
\lvv \vgamma - \vgamma^{\prime} \rvv_{2}^{2} = 
\lvv (\vgamma - \vgamma^{\prime\prime}) + 
(\vgamma^{\prime\prime} - \vgamma^{\prime}) \rvv_{2}^{2} 
\nonumber \\
&=& \lvv \vgamma - \vgamma^{\prime\prime} \rvv_{2}^{2} +
\lvv \vgamma^{\prime\prime} - \vgamma^{\prime} \rvv_{2}^{2} +
2 \langle \vgamma - \vgamma^{\prime\prime} , \vgamma^{\prime\prime} - \vgamma^{\prime} \rangle
\nonumber \\
&\ge& \lvv \vgamma - \vgamma^{\prime\prime} \rvv_{2}^{2}.
\label{epsilon_net_validity_chk}
\end{eqnarray}
The last inequality is obtained by dropping the last two nonnegative terms in the RHS. 
From \eqref{epsilon_net_validity_chk}, $\veca = \vgamma^{\prime\prime}$ will work.

Since case (i) and (ii) together are exhaustive, $\Theta^{\delta}(\setS)$ in step-$4$ is a 
valid $\delta$-net of $\Theta(\setS)$ which is entirely inside $\Theta(\setS)$.

\emph{Cardinality of $\Theta^{\delta}(\setS)$:}
The diameter of $\Theta(\setS)$ is $\sqrt{|\setS|}(\vgammamax - \vgammamin)$.
Based on the construction in step-4, the cardinality of $\Theta^{\delta}(\setS)$ 
can be upper bounded as: 
\begin{eqnarray}
|\Theta^{\delta}(\setS)| &\le& |\Theta^{\delta}_{\uparrow \delta}(\setS) \cap \Theta(\setS)| + |\Theta^{\delta}_{\uparrow \delta}(\setS) \cap \Theta(\setS)^{c}| 
\nonumber \\
&=&
|\Theta^{\delta}_{\uparrow \delta}(\setS)|
\le  
\lv \delta\text{-net of $\ell_{2}$ ball of radius $\sqrt{|\setS|}(\vgammamax - \vgammamin)$ in $\Real^{|\setS|}$}
\rv
\nonumber \\
&\le&
\max \lb 1, \lb 3 \sqrt{|\setS|}(\vgammamax - \vgammamin) / \delta \rb^{|\setS|} \rb.
\label{cardinality_epsilon_net}
\end{eqnarray}
The last step is an extension of the volumetric arguments in~\cite{Vershynin09RoleOfSparsity} to show that 
the $\delta$-covering number of a unit ball $\mathcal{B}_{1}(0)$ in $\Real^{k}$ with respect to the standard Euclidean norm 
$\lvv \cdot\rvv_{2}$ satisfies $\mathcal{N}_{\text{cov}}^{\delta} \lb \mathcal{B}_{1}(0), \lvv \cdot\rvv_{2}\rb \le \lb 3 / \delta \rb^{k}$. 
The $\texttt{max}$ operation with unity covers the case when $\delta$ is larger than the diameter of $\Theta(\setS)$ and the case when $\vgammamax = \vgammamin$.

Now consider the modified net $\Theta^{\epsilon / C_{\LL, \setS}}(\setS)$ obtained 
by setting $\delta = \frac{\epsilon}{ C_{\LL, \setS}}$ in steps $1$-$4$, 
where $C_{\LL, \setS}$ is the Lipschitz constant of $\LL(\matY, \vgamma)$ with respect to $\vgamma \in \Theta(\setS)$. We claim that $\Theta^{\epsilon/C_{\LL, \setS}}(\setS)$ is the desired set 
which simultaneously satisfies conditions (i) and (ii) stated in Proposition~\ref{prop_cov_num_enet}.  

To show condition (i), we observe that since $\Theta^{\epsilon/C_{\LL, \setS}}(\setS)$ is an $(\epsilon / C_{\LL, \setS})$-net 
of $\Theta(\setS)$ with respect to $\lvv\cdot\rvv_{2}$, for any $\vgamma \in \Theta(\setS)$, there exists a 
$\vgamma^{\prime} \in \Theta^{\epsilon/C_{\LL, \setS}}(\setS)$ such that 
$\lvv \vgamma - \vgamma^{\prime} \rvv_{2} \le \epsilon / C_{\LL, \setS}$.
Since $\LL(\matY, \vgamma)$ is $C_{\LL, \setS}$-Lipschitz in $\Theta(\setS)$, it follows that 
$\lv L(\matY, \vgamma) - L(\matY, \vgamma^{\prime}) \rv  \le C_{\LL, \setS} \lvv \vgamma - \vgamma^{\prime} \rvv_{2} \le 
\epsilon$.

Condition (ii) follows from \eqref{cardinality_epsilon_net} by setting 
$\delta = \epsilon / C_{\LL, \setS}$.

\subsection{Proof of Theorem \ref{thm_model_mismatch_ldp}} \label{app:thmproof_err_evt_as_ldp}
For continuous probability densities $p_{\vgamma}$ and $p_{\vgamma^{*}}$ defined on the observation 
space $\Real^{m}$, for any $\epsilon > 0$, the tail probability 
of the random variable $\log{\lb p_{\vgamma}(\matY) / p_{\vgamma^{*}}(\matY) \rb}$ has a Chernoff 
upper bound with parameter $t > 0$ as shown below. 
\begin{eqnarray}
	\PP \lb \log{\frac{p_{\vgamma}(\matY)}{p_{\vgamma^*}(\matY)}} \ge - \epsilon \rb &=&
	\PP \lb \sum_{j = 1}^{L} \log{\frac{p_{\vgamma}(\vecy_{j})}{p_{\vgamma^{*}}(\vecy_{j})}} \ge - \epsilon \rb  
	\nonumber \\
	&\le&  
	\EE_{p_{\vgamma^{*}}} \ls \exp{ \lb t \sum_{j = 1}^{L} \log{\frac{p_{\vgamma}(\vecy_{j})}{p_{\vgamma^{*}}(\vecy_{j})} }\rb}\rs  
	\exp{(t \epsilon)}
	\nonumber \\
	&=&  
	\lb \EE_{p_{\vgamma^{*}}} \ls \exp{ \lb t \log{\frac{p_{\vgamma}(\vecy)}{p_{\vgamma^{*}}(\vecy)} }\rb}\rs \rb^{L}  
	\exp{(t \epsilon)}
	\nonumber \\
	&=&  
	\lb \EE_{p_{\vgamma^{*}}} \ls \lb \frac{p_{\vgamma}(\vecy)}{p_{\vgamma^{*}}(\vecy)} \rb^{t} \rs \rb^{L}  
	\exp{(t \epsilon)}
	\nonumber \\
	&=&  
	\lb \int_{\vecy} p_{\vgamma}(\vecy)^{t} p_{\vgamma^{*}}(\vecy)^{1-t} d \vecy \rb^{L}  
	\exp{(t \epsilon)}
	\nonumber \\
	&=&  
	\exp{\lb -L \ls t \lb-\frac{\epsilon}{L}\rb - (t-1) \DD_{t}(p_{\vgamma}, p_{\vgamma^{*}}) \rs \rb}.
	\label{llr_tail_bnd_intr0}
\end{eqnarray}
In the above, the first and third steps follow from the independence of $\vecy_{j}$. The second step is 
the application of the Chernoff bound. The last step is obtained by using the definition of the R\'{e}nyi divergence 
and  rearranging the terms in the exponent.

By introducing the function $\psi(t) = (t-1)\DD_{t}(p_{\vgamma}, p_{\vgamma^{*}})$, the Chernoff bound in \eqref{llr_tail_bnd_intr0} can be restated as 
\begin{equation} \label{llr_tail_bnd_intr1}
	\PP \lb \log{\frac{p_{\vgamma}(\matY)}{p_{\vgamma^{*}}(\matY)}} \ge - \epsilon \rb 
	\le 
	\exp{\lb -L \ls t \lb-\frac{\epsilon}{L}\rb - \psi(t) \rs \rb}.
\end{equation}
For $t = \arg \sup_{t > 0} \lb t \lb-\frac{\epsilon}{L}\rb - \psi(t) \rb$, the upper bound in~\eqref{llr_tail_bnd_intr1} attains its 
tightest value \\ $\exp{\lb-L \psi^{*}\lb- \frac{\epsilon}{L} \rb\rb}$, where $\psi^{*}$ is the Legendre transform of $\psi$.

\subsection{Proof of Theorem~\ref{thm_abstract_suff_condition}}
\label{app:proof_thm_abstract_suff_condition}
Since $L \ge C \log{\frac{2}{\delta}}$, 
by Proposition \ref{prop_conc_spectralnorm_covmat}, $\PP(\setG^{c}) \le \delta$.
Combined with $L \ge \frac{8\kappa_{\text{cov}}}{\eta}$, \eqref{perr_form2} can be 
rewritten as  
\begin{equation}
\PP(\mathcal{E}_{\setS^{*}}) \le 
\sum_{\setS \in \setS_{n} \backslash \setS^{*}} 
\exp{\lb  - \frac{\eta L k_{d}^{\setS, \setS^{*}}}{8} \rb} + \delta.
\label{perr_form3}
\end{equation}
The total number of support sets belonging to $\setS_{n} \backslash \setS^{*}$ which differ from 
the true support $\setS^{*}$ in exactly $k_{d}$ locations is 
$\sum_{j = 0}^{k_{d}}{n - |\setS^{*}|\choose j} 
{j + |\setS^{*}| \choose \min\lb k_{d}, j + |\setS^{*}| \rb} $. Since $|\setS^{*}| \le K$, this summation can be further upper bounded by $\lb 2nK \rb^{k_{d}}$. 
Thus, we can rewrite \eqref{perr_form3} as 
\begin{eqnarray} 
\PP(\mathcal{E}_{\setS^{*}}) &\le& 
\delta \;\; + 
\sum_{k_d =1}^{n-|\setS^{*}|} \!\!
\sum_{\substack{\setS \in \setS_{n} \backslash \setS^{*}, \\ 
|(\setS \backslash \setS^{*}) \cup (\setS^{*} \backslash \setS)|=k_{d} }}
\!\!\!\!\!\! \exp{\lb -\frac{\eta L k_{d}}{8}  \rb} 
\nonumber \\
&\le& \delta + 
\sum_{k_d =1}^{n-|\setS^{*}|} 
\lb 2nK\rb^{k_{d}}  \lb e^{-\frac{\eta L }{8}} \rb^{k_{d}}.
\label{perr_proof_intr1}
\end{eqnarray}
Since $L \ge \frac{8}{\eta} \log{\lb 2n K \lb \frac{1+\delta}{\delta} \rb  \rb}$, $\PP(\mathcal{E}_{\setS^{*}})$ can be 
upper bounded by a geometric series as 
\begin{eqnarray}
\PP(\mathcal{E}_{\setS^{*}}) &\le&  \delta + 
\sum_{k_d =1}^{\infty} \lb \frac{\delta}{1+\delta} \rb^{k_d}
\;=\;\; \delta + \delta \;=\; 2\delta.
\nonumber
\end{eqnarray}

\subsection{Proof of Proposition \ref{renyi_div_lb_gaussian_case}} \label{app:renyi_div_lb_gaussian_case}
Let $\Delta \Gamma = \diag(\Delta \vgamma)$, where $\Delta \vgamma \triangleq \vgamma - \vgamma^{*}$. Also, let $\matSigma_{\vgamma + \vgamma^*} \triangleq \matSigma_{\vgamma} + \matSigma_{\vgamma^*}$. Then, using Proposition~\ref{prop_renyi_div_lb_analytical}, $\DD_{\frac{1}{2}}(p_{\vgamma}, p_{\vgamma^*})$ can be bounded as follows.
\begin{eqnarray}
\DD_{\frac{1}{2}}(p_{\vgamma}, p_{\vgamma^*}) &\ge&
\frac{1}{2} \text{tr}\lb \matSigma_{\vgamma + \vgamma^*}^{-1} (\matA \Delta \Gamma \matA^{T}) \matSigma_{\vgamma + \vgamma^*}^{-1} (\matA \Delta \Gamma \matA^{T}) \rb
\nonumber \\
&\ge&
\frac{\lvv \matA \Delta \Gamma \matA^{T} \rvv_{F}^{2}}{2\lvvv \matSigma_{\vgamma + \vgamma^*} \rvvv_{2}^{2}} = 
\frac{\lvv (\matA \odot \matA) \Delta \vgamma \rvv_{2}^{2}}{2\lvvv  \matSigma_{\vgamma + \vgamma^*} \rvvv_{2}^{2}}.
\label{eqn_traceXYXY_lb}
\end{eqnarray}
In above, the second inequality is obtained by repeatedly applying the trace inequality: 
$\text{tr}\lb \matA^{-1} \matB\rb$ $ \ge$ $ \text{tr}(\matB)/\lvvv \matA \rvvv_{2}$ for any positive definite 
$\matA$ and positive semidefinite~$\matB$. The last step follows from the identity: $\text{vec}(\matA \Delta \Gamma \matA^{T}) = (\matA \odot \matA) \Delta \vgamma$. 


Next, we derive an upper bound for the spectral norm of~$\matSigma_{\vgamma + \vgamma^*}$ as shown below. 
\begin{eqnarray}
\lvvv \matSigma_{\vgamma + \vgamma^*} \rvvv_{2} &=& 
\lvvv 2\sigma^{2}\matI_{m} + \matA (\matGamma + \matGamma^*) \matA^{T} \rvvv_{2} 
\nonumber \\
&\le& 2\sigma^{2} + 2 \vgammamax \lvvv \matA_{\setS \cup \setS^*}^{T} \matA_{\setS \cup \setS^*} \rvvv_{2}.
\label{eqn_rendiv_lb_intr3}
\end{eqnarray}

Finally, using \eqref{eqn_rendiv_lb_intr3} in \eqref{eqn_traceXYXY_lb}, we obtain the desired lower bound for $\DD_{\frac{1}{2}}(p_{\vgamma}, p_{\vgamma^*})$. 

\ifdefined \SKIPTEMP
Let $\Delta \Gamma = \diag(\Delta \vgamma)$, where $\Delta \vgamma \triangleq \vgamma - \vgamma^{*}$. Then, using Proposition~\ref{prop_renyi_div_lb_analytical}, $\DD_{1/2}(p_{\vgamma}, p_{\vgamma^*})$ can be bounded as follows.
\begin{eqnarray}
\DD_{\frac{1}{2}}(p_{\vgamma}, p_{\vgamma^*}) 
& \ge&
\frac{c_{1}}{4} \text{tr}\lb (\matA \Delta \Gamma \matA^{T}) \matSigma_{\vgamma^*}^{-2} (\matA \Delta \Gamma \matA^{T}) \rb
\nonumber \\
&& \hspace{-2cm}  \ge
\frac{c_{1}\lvv \matA \Delta \Gamma \matA^{T} \rvv_{F}^{2}}{4\lvvv \matSigma_{\vgamma^*} \rvvv_{2}^{2}} = 
\frac{c_{1}\lvv (\matA \odot \matA) \Delta \vgamma \rvv_{2}^{2}}{4\lvvv  \matSigma_{\vgamma^*} \rvvv_{2}^{2}}.
\label{eqn_traceXYXY_lb}
\end{eqnarray}
In above, the first step is due to symmetric nature of $\DD_{\alpha}$ for $\alpha = 1/2$.
The second inequality is obtained by applying the trace inequality: 
$\text{tr}\lb \matA^{-1} \matB\rb \ge \text{tr}(\matB)/\lvvv \matA \rvvv_{2}$ for any positive definite 
$\matA$ and positive semidefinite~$\matB$. The last step follows from the identity: $\text{vec}(\matA \Delta \Gamma \matA^{T}) = (\matA \odot \matA) \Delta \vgamma$. 

Further, using the fact that $\lvv \matSigma_{\gamma^*}\rvv_{2} \le  \sigma^2 + \vgammamax \sigma_{\text{max}}^{2}(\matA_{\setS^*})$  in \eqref{eqn_traceXYXY_lb}, we obtain the desired lower bound for $\DD_{1/2}(p_{\vgamma}, p_{\vgamma^*})$. 
\fi

\subsection{Proof of Theorem~\ref{thm_strong_rnsp_kr}} 
\label{app:proof_thm_strong_rnsp_kr}

Consider the unit norm $m^2$ length vector 
\begin{equation}
\vecw \triangleq \frac{\mathbf{1}_{\lc 1, m+2, 2m+3, \ldots, m^2 \rc}}{\sqrt{m}},
\end{equation}
where $\mathbf{1}_{\setS}$ is a binary vector containing ones in the indices specified by the set $\setS$ and zeros everywhere else. 
We call $\vecw$ the \textit{Hadamard sampler}, as it samples the $m$ rows of the Hadamard submatrix contained within $\matA \odot \matA$. Let $\vecb = 
(\matA \odot \matA)^{T} \vecw$, then 
\begin{equation}
\vecb(i) = \frac{ (\veca_{i} \circ \veca_{i})^{T} \mathbf{1}_{[m]}}{\sqrt{m}} = \frac{\lvv \veca_{i} \rvv_{2}^{2}}{\sqrt{m}} \;\; \forall i \in [n],
\end{equation}
where $\veca_{i}$ denotes the $i^{\text{th}}$ column of $\matA$. Since we have assumed that $\lvv \veca_{i}\rvv_{2}^{2} \in \ls 1 - \alpha, 1+ \alpha \rs$,  
\begin{equation}
\frac{1-\alpha}{\sqrt{m}} \mathbf{1}_{[n]} \preceq (\matA \odot \matA)^{T} \vecw \preceq \frac{1+\alpha}{\sqrt{m}} \mathbf{1}_{[n]}. 
\label{kr_robust_nsp_intr1}
\end{equation}

To ease the notation, let $\matX = \matA \odot \matA$. Given projection matrices $\vecw \vecw^{T}$ and $\Pi = \matI_{m^2} - \vecw \vecw^{T}$, one can write 
\begin{eqnarray}
\vecv^{T} \matX^{T} \matX \vecv  &=& \vecv^{T}\matX^{T} \vecw \vecw^{T}\matX \vecv + 
\vecv^{T}\matX^{T} \Pi \matX \vecv
\nonumber \\
&\ge&
\vecv^{T}\matX^{T} \vecw \vecw^{T}\matX \vecv
\nonumber \\
&=& 
\vecv_{+}^{T}\matX^{T} \vecw \vecw^{T}\matX \vecv_{+}
+ \vecv_{-}^{T}\matX^{T} \vecw \vecw^{T}\matX \vecv_{-}
- 2 \vecv_{+}^{T}\matX^{T} \vecw \vecw^{T}\matX \vecv_{-}
\nonumber \\
& \stackrel{(a)}{\ge}& 
\frac{(1-\alpha)^{2}}{m} \lb \vecv_{+}^{T} \mathbf{1}_{n} \mathbf{1}_{n}^{T} \vecv_{+} \rb
+ \frac{(1-\alpha)^{2}}{m} \lb \vecv_{-}^{T}\mathbf{1}_{n} \mathbf{1}_{n}^{T} \vecv_{-} \rb
- 2 \frac{(1+\alpha)^{2}}{m} \lb \vecv_{+}^{T} \mathbf{1}_{n} \mathbf{1}_{n}^{T} \vecv_{-} \rb
\nonumber \\
&=& 
\frac{(1 \!-\! \alpha)^2}{m} \lb \lvv \vecv_{+} \rvv_{1}^{2} +  \lvv \vecv_{-} \rvv_{1}^{2} \rb
- 2 \lvv \vecv_{+} \rvv_{1} \lvv \vecv_{-} \rvv_{1} \frac{(1 \!+\! \alpha)^2}{m}
\nonumber \\
&=& 
\frac{(1 - \alpha)^2}{m} \lb \lvv \vecv_{+} \rvv_{1}^{2} +  \lvv \vecv_{-} \rvv_{1}^{2} \rb
\ls 1 - \frac{2 \lvv \vecv_{+} \rvv_{1} \lvv \vecv_{-} \rvv_{1}}{\lvv \vecv_{+} \rvv_{1}^{2} +  \lvv \vecv_{-} \rvv_{1}^{2}} \lb \frac{1 + \alpha}{1-\alpha} \rb^{2}
\rs.
\label{kr_robust_nsp_intr2}
\end{eqnarray}
In above, step ($a$) follows from \eqref{kr_robust_nsp_intr1} and the nonnegativity of $\vecv_{+}$ and $\vecv_{-}$. 
We observe that for $\lvv \vecv_{+} \rvv_{1} > 4 \lb \frac{1 + \alpha}{1 - \alpha}\rb^{2} \lvv \vecv_{-} \rvv_{1}$, the ratio 
$\frac{2 \lvv \vecv_{+} \rvv_{1} \lvv \vecv_{-} \rvv_{1}}{\lvv \vecv_{+} \rvv_{1}^{2} +  \lvv \vecv_{-} \rvv_{1}^{2}} \le \frac{1}{2} \lb \frac{1 - \alpha}{1 + \alpha}\rb^{2}$, and therefore
\begin{equation}
\vecv^{T}\matX^{T} \matX \vecv \ge \frac{(1 - \alpha)^2}{2m} \lb \lvv \vecv_{+} \rvv_{1}^{2} +  \lvv \vecv_{-} \rvv_{1}^{2} \rb.
\nonumber
\end{equation}

\subsection{Proof of Proposition~\ref{prop_eta_lb}} \label{app:proof_prop_eta_lb}
\begin{proof} 
Let us define $\Delta \vgamma = \vgamma - \vgamma^*$ which also splits as 
\begin{equation}
\Delta \vgamma = \Delta \vgamma_{+} - \Delta \vgamma_{-}, 
\label{dgamma_pm_decompose}
\end{equation}
where $\Delta \vgamma_{+}$ and $\Delta \vgamma_{-}$ are nonnegative vectors in $\Real^{n}_{+}$ 
with non-overlapping supports and containing absolute values of positive and negative coefficients of $\Delta \vgamma$, respectively. Let $\setS$ and $\setS^*$ denote the nonzero supports of 
$\vgamma$ and $\vgamma^*$, respectively. Suppose $\setS$ and $\setS^*$ differ in exactly $k_{d}^{\setS, \setS^*}$ locations.  
By construction of $\Delta \vgamma_{+}$ and $\Delta \vgamma_{-}$, we have 
\begin{eqnarray}
& \lvv \Delta \vgamma \rvv_{2}^{2} &\ge k_{d}^{\setS, \setS^*} \vgammamin^2
\label{dgamma_bnds} \\
\lb |\setS^*| - k_{d}^{\setS, \setS^*}\rb_{+} \vgammamin  \le & \lvv \Delta \vgamma_{-} \rvv_{1} &
\le |\setS^*| \vgammamax 
\label{dgamma_minus_bnds} \\
\lb k_{d}^{\setS, \setS^*} - |\setS^*| \rb_{+} \vgammamin  \le & \lvv \Delta \vgamma_{+} \rvv_{1} &
 \le k_{d}^{\setS, \setS^*} \vgammamax + |\setS^*| (\vgammamax - \vgammamin)
\label{dgamma_plus_bnds} 
\end{eqnarray} 
We introduce $K_{\text{threshold}} \triangleq \lb 1 + 4 \frac{\vgammamax}{\vgammamin} \lb \frac{1 + \alpha}{1- \alpha} \rb^{2} \rb K$, and $\setB \triangleq \lc \setS \in [n]: k_{d}^{\setS, \setS^*} \le K_{\text{threshold}} \rc$. 
Then, from \eqref{eta_char}, we have
\begin{align}
\eta = \min_{\setS \subseteq [n]} \frac{\DD_{\setS}^*}{k_{d}^{\setS, \setS^*}} 
 = \min \lb  \min_{\setS \in \setB } \frac{\DD_{\setS}^*}{k_{d}^{\setS, \setS^*}}, \;
\min_{\setS \in \setB^{c} } \frac{\DD_{\setS}^*}{k_{d}^{\setS, \setS^*}} \rb.
\label{eqn_eta_lb_split}
\end{align}
Note that for $\text{supp}(\vgamma) = \setS$ and $\setS \in \setB$, from property $P2$, we have $\lvv (\matA \odot \matA) \Delta \vgamma \rvv_{2}^{2} \ge \beta \lvv \Delta \vgamma \rvv_{2}^{2}$. Using the lower bound on $\DD_{\setS}^*$ derived in Proposition~\ref{renyi_div_lb_gaussian_case}, we can bound $\min_{\setS \in \setB } \frac{\DD_{\setS}^*}{k_{d}^{\setS, \setS^*}}$ as follows.
\begin{eqnarray}
\min_{\setS \in \setB } \frac{\DD_{\setS}^*}{k_{d}^{\setS, \setS^*}} &=&
\min_{\substack{\vgamma \in \Theta(\setS), \\ \setS \in \setB } } \;
\frac{ \beta \lvv \Delta \vgamma \rvv_{2}^{2}}
{4 k_{d}^{\setS, \setS^*} \lb \sigma^2 +  \vgammamax \sigma^2_{\text{max}}(\matA_{\setS \cup \setS^*}) \rb^{2} } 
\nonumber \\
&\ge& \min_{\setS \in \setB } 
\frac{ \beta k_{d}^{\setS, \setS^*} \vgammamin^2 }
{4 k_{d}^{\setS, \setS^*} \lb \sigma^2 +  \vgammamax \sigma^2_{\text{max}}(\matA_{\setS \cup \setS^*}) \rb^{2} } 
\nonumber \\
&\ge&  
\frac{\beta \vgammamin^2 }{4 \lb \sigma^2 +  \vgammamax \rb^2 \max(1, \delta_{(K + K_{\text{threshold}})}^{2}) }. 
\label{proof_eta_lb_intr3}
\end{eqnarray}

For the case where $\setS \in \setB^c$, i.e., $k_{d}^{\setS, \setS^*} > K_{\text{threshold}}$, it follows from \eqref{dgamma_minus_bnds} and \eqref{dgamma_plus_bnds} that $\lvv \Delta \vgamma_{+} \rvv_{1} \ge 
4 \lb \frac{1 + \alpha}{1 - \alpha} \rb \lvv \Delta \vgamma_{-} \rvv_{1}$. Therefore, we can invoke the 
restricted null space property of $\matA \odot \matA$ from 
Theorem~\ref{thm_strong_rnsp_kr} to bound $\min_{\setS \in \setB^{c} } 
\frac{\DD_{\setS}^*}{k_{d}^{\setS, \setS^*}}$ as follows. 
\begin{eqnarray}
\min_{\setS \in \setB^{c} } \frac{\DD_{\setS}^*}{k_{d}^{\setS, \setS^*}} 
&\ge&
\min_{\substack{\vgamma \in \Theta(\setS), \\ \setS \in \setB^{c} } }
\frac{(1- \alpha)^2 \lb \lvv \Delta \vgamma_{+}\rvv_{1}^{2} + \lvv \Delta \vgamma_{-}\rvv_{1}^{2} \rb}
{8 m k_{d}^{\setS, \setS^*} \lb \sigma^2 +  \vgammamax \sigma^2_{\text{max}}(\matA_{\setS \cup \setS^*}) \rb^{2} } 
\nonumber \\
&\ge&
\min_{\substack{\vgamma \in \Theta(\setS), \\ \setS \in \setB^{c} } }
\frac{ (1 - \alpha)^2 \lb \lvv \Delta \vgamma_{+}\rvv_{1}^{2}  \rb }
{8 m k_{d}^{\setS, \setS^*} \!\!\lb \sigma^2 \!+\! \vgammamax \rb^2 \lb \max \lb 1,  \sigma^2_{\text{max}}(\matA_{\setS \cup \setS^*}) \rb \rb^{2} } 
\nonumber \\
&\ge&
\min_{\setS \in \setB^{c} }
\frac{ (1-\alpha)^2 \vgammamin^{2} \lb k_{d}^{\setS, \setS^*} - |\setS^*| \rb^{2} }
{8 m k_{d}^{\setS, \setS^*} \lb \sigma^2 + \vgammamax \rb^2 \lb \max \lb 1,  
\sigma^2_{\text{max}}(\matA_{\setS \cup \setS^*}) \rb \rb^{2} } 
\nonumber 
\label{proof_ds_lb} \\
&\ge&
\frac{(1- \alpha)^2 \vgammamin^2}{8 m \lb \sigma^2 +  \vgammamax \rb^2}
\min_{\setS \in \setB^{c}} \lb 1 - \frac{|\setS^*|}{ k_{d}^{\setS, \setS^*} } \rb 
\lb \min_{\setS \in \setB^{c}}
\frac{ k_{d}^{\setS, \setS^*} - |\setS^*| }{\lb \max \lb 1,  \sigma^2_{\text{max}}(\matA_{\setS \cup \setS^*}) \rb \rb^{2}} \rb
\nonumber \\
&\ge&
\frac{(1-\alpha)^2 \vgammamin^2}{8 m \lb \sigma^2 +  \vgammamax \rb^2}
\lb 1 - \frac{K}{ K_{\text{threshold}} } \rb 
\lb \min_{\setS \in \setB^{c}}
\frac{ k_{d}^{\setS, \setS^*} - |\setS^*| }{\lb \max \lb 1,  \sigma^2_{\text{max}}(\matA_{\setS \cup \setS^*}) \rb \rb^{2}} \rb
\nonumber \\
&\ge&
\frac{(1-\alpha)^2 \vgammamin^2}{10 m \lb \sigma^2 + \vgammamax \rb^2 }
\lb \min_{\setS \in \setB^{c}}
\frac{k_{d}^{\setS, \setS^*} - |\setS^*|}{|\setS \cup \setS^*|} \rb
\lb \min_{\setS \in \setB^{c}}
\frac{|\setS \cup \setS^*|}{\lb \max \lb 1,  \sigma^2_{\text{max}}(\matA_{\setS \cup \setS^*}) \rb\rb^{2}}
\rb
\nonumber \\
&=&
\frac{ (1-\alpha)^2 \vgammamin^2}{10 m \lb \sigma^2 +  \vgammamax \rb^2 }
\lb \min_{\setS \in \setB^{c}}
\frac{k_{d}^{\setS, \setS^*} - |\setS^*|}{k_{d}^{\setS, \setS^*} + |\setS^*|} \rb
\lb \min_{\setS \in \setB^{c}}
\frac{|\setS \cup \setS^*|}{\lb \max \lb 1,  \sigma^2_{\text{max}}(\matA_{\setS \cup \setS^*}) \rb\rb^{2}}
\rb
\nonumber \\
&\ge&
\frac{(1-\alpha)^2 \vgammamin^2}{16 m \lb \sigma^2 +  \vgammamax \rb^2 }
\lb \min_{\setS \in \setB^{c}}
\frac{|\setS \cup \setS^*|}{\max (1, \delta^{2}_{|\setS \cup \setS^*|}) }
\rb.
\label{proof_eta_lb_intr4}
\end{eqnarray}
\textcolor{black}{In the above, the penultimate inequality follows by substituting $K_{\text{threshold}}$ with its lower bound $5K$.}
Substituting \eqref{proof_eta_lb_intr3} and \eqref{proof_eta_lb_intr4} in 
\eqref{eqn_eta_lb_split} and simplifying, we obtain the lower bound for $\eta$ stated in the proposition.
\ifdefined \SKIPTEMP
\begin{equation}
\eta \! \ge\! 
\frac{\vgammamin^2}{\lb \sigma^2 \!+\! 2 \vgammamax \rb^{2}}  
\min \lb \! \!
\frac{\beta}{\delta^{2}_{(\!K \!+\! K_{\text{threshold}\!})\!} },\! 
\frac{(1-\alpha)^2\!}{4m \delta_{n}} \!\min_{\setS \in \setB^{c}}
\frac{|\setS \cup \setS^*|}{\delta_{|\setS \cup \setS^*|}}
\rb \!\!.
\nonumber \\
\label{proof_eta_lb_intr5}
\end{equation}
This concludes our proof.
\fi
\end{proof}

\subsection{Proof of Proposition~\ref{prop_kcov_ub}} \label{proof_prop_kcov_ub}
\begin{proof}
For any support $\setS \subset [n]$, by setting $\epsilon = \frac{L \DD_{\setS}^{*}}{2}$ in Proposition~\ref{prop_cov_num_enet}, we have 
\begin{equation} \label{max_covering_num}
|\Theta^{\epsilon}(\setS) \vert_{\setG} | \le 
\max \lc 1, 
\lb  \frac{6 \sqrt{|\setS|} (\vgammamax - \vgammamin) C_{\LL, \setS}}{L 
\DD_{\setS}^{*}}  \rb^{|\setS|} \rc.
\end{equation} 
where $C_{\LL, \setS}$ denotes the Lipschitz constant of $\LL(\matY; \vgamma)$ with respect to 
$\vgamma$ over the bounded domain $\Theta^{\epsilon}(\setS) \vert_{\setG}$. 
Proposition~\ref{loglikelihood_Lipschitz_char} characterizes the Lipschitz property of $\LL(\matY, \vgamma)$. 
\begin{proposition} \label{loglikelihood_Lipschitz_char}
For $\setS \in \setS_{n}$, the log-likelihood $\LL(\matY; \vgamma): \Theta(\setS) \to~\Real$ 
is Lipschitz continuous in $\vgamma$ as shown below.
\begin{align}
\lv \LL(\matY, \vgamma_{2}) - 
\LL(\matY, \vgamma_{1}) \rv \le
\frac{m L}{\vgammamin} 
\lb 1 + \frac{\lvvv \matR_{\vecy\vecy} \rvvv_{2}}{\sigma^2} \rb 
\lvv \vgamma_{2} - \vgamma_{1} \rvv_{2},
\nonumber 
\end{align}
for any $\vgamma_{1}, \vgamma_{2} \in \Theta(\setS)$. Here, $\matR_{\vecy\vecy} \triangleq \frac{1}{L} \matY \matY^{T}$.
\end{proposition}
\begin{proof}
	See Appendix \ref{app:proof_loglikelihood_Lipschitz_char}.
\end{proof}
By invoking the definitions of $K_{\text{threshold}}$ and set $\setB$ from the proof of Proposition~\ref{prop_eta_lb}, we can rewrite $\kappa_{\text{cov}}$ in \eqref{epsilon_cover_cardinality_bound} as  
\begin{equation}
\kappa_{\text{cov}} = \max \lb  
\max_{\setS \in \setB}  \frac{ \log \big| \Theta^{\epsilon}(\setS) \vert_{\setG} \big|}{k_{d}^{\setS, \setS^*}}, 
 \max_{\setS \in \setB^c}  \frac{ \log \big| \Theta^{\epsilon}(\setS) \vert_{\setG} \big|}{k_{d}^{\setS, \setS^*}}
\rb.
\label{eqn_kcov_split}
\end{equation}

For the MMV set $\setG$ as defined in \eqref{defn_bdd_spectral_norm_mmv_set} 
and $\matY \in 
\setG$, we have 
$\lvvv \matR_{\vecy\vecy}\rvvv_{2} \le 2 \lb \sigma^2 + \vgammamax \delta_{K} 
\rb$. Then, for $\vgammamin \neq \vgammamax$ and by the Lipschitz continuity of 
$\LL(\matY; \vgamma)$ as per Proposition~\ref{loglikelihood_Lipschitz_char}, it 
follows that 
\begin{align}
\max_{\setS \in \setB}  \frac{ \log \big| \Theta^{\frac{L \DD_{\setS}^*}{2}}(\setS) \vert_{\setG} \big|}{k_{d}^{\setS, \setS^*}} 
&\le
\max_{\setS \in \setB} \frac{|\setS|}{k_{d}^{\setS, \setS^*}} 
\log \lb \frac{6 \sqrt{|\setS|} C_{\LL, \setS} (\vgammamax - \vgammamin)}{L 
\DD^{*}_{\setS}} \rb
\nonumber \\
&\le 
\max_{\setS \in \setB} \frac{|\setS^*| + k_{d}^{\setS, \setS^*}}{k_{d}^{\setS, \setS^*}} 
 \log \lb \frac{6 m \sqrt{|\setS^*| + k_{d}^{\setS, \setS^*}} 
(\vgammamax \!-\! \vgammamin) (3 + 2 \frac{\vgammamax}{\sigma^2} \delta_{K}) }{ \vgammamin 
\DD_{\setS}^*} \rb
\nonumber \\
& \le 
(K+1) \log{\lb m \sqrt{K+1} \Delta_{1} \rb} ,
\label{eqn_kcov_ub_intr1}
\end{align}
where $\Delta_{1} \!\! =\!\!  \frac{24 (\vgammamax - \vgammamin) \lb \! 3 +\!  
\frac{2\vgammamax \delta_{K}}{\sigma^2} \! \rb (\sigma^2 \! + \vgammamax )^{2} \!\max \lb \! 1, 
\delta_{K + K_{\text{threshold}}}^{2} \!\rb }{\beta \vgammamin^3}$. The last inequality is 
obtained by using \eqref{proof_eta_lb_intr3} and noting that the RHS in the 
second 
inequality decreases monotonically with respect to $k_{d}^{\setS, \setS^*}$ in 
the interval $[1, K_{\text{threshold}}]$.

The second max-term in \eqref{eqn_kcov_split} can be bounded as follows. By invoking Proposition~\ref{renyi_div_lb_gaussian_case}, we first note that  
\begin{align}
\min_{\setS \in \setB^c}\DD_{\setS}^{*} & \ge 
\min_{\substack{\vgamma \in \Theta(\setS), \\ \setS \in \setB^c}}
\frac{\lvv (\matA \odot \matA) (\vgamma - \vgamma^*) \rvv_{2}^{2}}{4 (\sigma^2 + \vgammamax 
	\sigma_{\max}^{2}(\matA_{\setS \cup \setS^*}))^{2}}
\nonumber \\
& \hspace{0cm} \ge  
\min_{\setS \in \setB^c}
\frac{(1-\alpha)^2 (k_{d}^{\setS, \setS^*} - |\setS^*|)^{2} \vgammamin^{2} }
{8 m (\sigma^2 + \vgammamax \sigma_{\max}^{2}(\matA_{\setS \cup \setS^*}))^{2}}
\nonumber \\
& \hspace{0cm} \ge  
\min_{\setS \in \setB^c}
\lb \frac{(k_{d}^{\setS, \setS^*}\! \!-\! |\setS^*|)^{2}}{|\setS \cup 
\setS^*|^{2}} 
\rb 
\frac{(1-\alpha)^2 \vgammamin^{2} }{8m (1 + \alpha)^{2} (\sigma^2 + \vgammamax)^2 }
\nonumber \\
& \hspace{0cm} = \min_{\setS \in \setB^c}  
\lb \frac{k_{d}^{\setS, \setS^*} - |\setS^*|}{k_{d}^{\setS, \setS^*} + 
|\setS^*|} \rb^{2} \!\!
\lb \frac{1 - \alpha}{1 + \alpha} \rb^{2} \!\! \frac{ \vgammamin^{2} }{8m (\sigma^2 + \vgammamax)^2 }
\nonumber \\
& \hspace{0cm} \ge   
\frac{1}{18 m} \lb \frac{1 - \alpha}{1 + \alpha} \rb^{2} \frac{\vgammamin^{2} }{(\sigma^2 + \vgammamax)^2 }.
\label{eqn_kcov_ub_intr2}
\end{align}
\textcolor{black}{The last inequality follows from the fact that $|\setS^*| \le K$ and the ratio $\lb \frac{k_{d}^{\setS, \setS^*} - |\setS^*|}{k_{d}^{\setS, \setS^*} + |\setS^*|} \rb$ is nondecreasing in $|\setS|$ for $\setS \in \setB^\text{c}$.}  
Then, substituting the lower bound for $\DD_{\setS}^*$ from 
\eqref{eqn_kcov_ub_intr2} in \eqref{max_covering_num}, and simplifying, we obtain 
\begin{align}
\max_{\setS \in \setB^c}  \frac{ \log \big| \Theta^{\frac{L 
\DD_{\setS}^*}{2}}(\setS) \vert_{\setG} \big|}{k_{d}^{\setS, \setS^*}}
&  
\le \max_{\setS \in \setB^c}  \frac{|\setS^*| + k_{d}^{\setS, \setS^*}}{k_{d}^{\setS, \setS^*}}
\log \lb m^2 \sqrt{|\setS^*| + k_{d}^{\setS, \setS^*}} \Delta_{2} \rb
\nonumber \\
& 
\le 4 \log m + 2 \log \Delta_{2} + \frac{n}{n-K} \log n
\label{eqn_kcov_ub_intr3}
\end{align}
where $\Delta_{2} = \frac{108 (\vgammamax - \vgammamin) \lb 3 + 2 
\frac{\vgammamax \delta_{K}}{\sigma^2} \rb (\sigma^2 +  \vgammamax)^{2} (1 + \alpha)^2 } { 
(1- \alpha)^2 \vgammamin^3}$. The coefficient $\frac{n}{n-K}$ in \eqref{eqn_kcov_ub_intr3} can be replaced by a suitable constant, say $2$, even for the worst scaling $K = \Theta(n)$.

Finally, substituting the bounds from \eqref{eqn_kcov_ub_intr1} and \eqref{eqn_kcov_ub_intr3} in~\eqref{eqn_kcov_split}, we obtain the claimed upper bound for $\kappa_{\text{cov}}$.
\end{proof}


\subsection{Proof of Corollary~\ref{corr_gaussian_meas_map_eta_char}}
\label{app:proof_corr_gaussian_meas_map_eta_char}
In order to bound $\eta$, we first derive probabilistic bounds for the parameters $\alpha$, $\beta$ and $\delta_{K}$ associated with the sensing matrix $\matA$. We consider both cases: i. $m = \Theta(K \log n)$ and ii. $m = \Theta(\sqrt{K} \log n)$.

\subsubsection{Lower bound for $\alpha$}
By invoking the Hanson-Wright concentration inequality \cite{Rudelson13HansonWright}, 
and taking the union bound over all columns of~$\matA$, 
\begin{align}
\PP \lb \bigcup_{i \in [n]} \lc \lv \lvv \veca_{i} \rvv_{2}^{2} - 1 \rv \ge \alpha \rc \rb 
\le n \PP \lb \lv \textcolor{black}{\veca_{1}^{T} \matI_m \veca_{1}} - 1 \rv \ge \alpha \rb 
\le 2n e^{-c m \alpha^2},
\label{eqn_random_A_results_intr1}
\end{align}
where $c$ is a numerical constant.
For both $m = \Theta(K \log n)$ and $m = \Theta(\sqrt{K} \log n)$, by setting $\alpha = \frac{1}{2}$ in  \eqref{eqn_random_A_results_intr1}, the squared $\ell_{2}$-norm of columns of $\matA$ lie inside the interval $\ls \frac{1}{2}, \frac{3}{2} \rs$ with probability exceeding $1 - n^{-\Theta(\sqrt{K})}$.

\subsubsection{Lower bound for $\beta$} 
In Proposition~\ref{prop_eta_lb}, $\beta$ refers to the smallest singular value among 
all submatrices obtained by sampling $K + K_{\text{threshold}}$ or fewer columns of $\matA \odot \matA$.
Noting that $K_{\text{threshold}} = \Theta(K)$, and   
\textcolor{black}{by invoking \cite[Theorem~$1$]{Khanna12KRRIC_suppl} (a revised version of \cite[Theorem~$5$]{alex19RIPKR}), we can have}
\begin{equation}
\PP \lb \beta < 1 - \frac{c_{2} \sqrt{K} \log n}{m} 
\rb \le \frac{c_{3}}{n^{2}}, 
\label{eqn_random_A_probbound_for_beta_2}
\end{equation}
where $c_{2}$ and $c_{3}$ are positive numerical constants. Thus, for both $m = \Theta(K \log n)$ and $m = \Theta(\sqrt{K} \log n)$, we have 
\begin{equation}
\beta = \Omega(1),
\label{eqn_random_A_results_intr0a}
\end{equation}
with probability exceeding $1 - \Omega(n^{-2})$.

\subsubsection{Upper bounds for $\delta_{K + K_{\text{threshold}}}$ and $\delta_{n}$}
From Proposition~\ref{prop_eta_lb}, $\delta_{k}$ is defined as the maximum squared singular value 
among any $k$ or fewer column submatrix of $\matA$. Thus, by direct application of Corollary~\ref{corr_Gaussian_submatrix_spectral_bound}, for $k \le n$, 
\begin{equation} 
\delta_{k} \le \frac{\lb \sqrt{m} + \sqrt{k} + \sqrt{6 k \log n}  \rb^{2}}{m}
\label{eqn_random_A_results_intr2}
\end{equation}
with probability exceeding $1 - 2n^{-k}$. From \eqref{eqn_random_A_results_intr2}, for $m = \Theta(K \log n)$ we have   
\begin{equation}
\delta_{K + K_{\text{threshold}}} = O(1), \text{ and } \delta_{n} = O\lb \frac{n}{K \log n } \rb,
\label{eqn_random_A_deltak_ub}
\end{equation}
and $m = \Theta(\sqrt{K} \log n)$ guarantees that 
\begin{equation}
\delta_{K + K_{\text{threshold}}} = O(\sqrt{K}), \text{ and } \delta_{n} = O\lb \frac{n}{\sqrt{K} \log n } \rb,
\label{eqn_random_A_deltak_ub2}
\end{equation}
with probability exceeding $1-n^{-\Theta(\sqrt{K})}$. In the above, the bounds for $\delta_{n}$ follow from Proposition~\ref{prop_gaussian_spectral_bound}.

Finally, for $m = \Theta(K \log n)$, by substituting the above bounds for $\alpha$, $\beta$, $\delta_{K}$ and $\delta_{n}$ from \eqref{eqn_random_A_results_intr0a} and \eqref{eqn_random_A_deltak_ub} in Propositions~\ref{prop_eta_lb} and \ref{prop_kcov_ub}, and simplifying, it can be verified that $\eta = \Omega  \lb \frac{K}{n} \rb$ and $\kappa_{\text{cov}} \le O(K \log K + K \log \log n + \log n)$, respectively, 
with probability exceeding $1 - \Omega(n^{-2})$. Likewise, for the $m = \Theta(\sqrt{K} \log n)$ case, we have $\eta = \Omega  \lb \frac{\sqrt{K}}{n} \rb$ and $\kappa_{\text{cov}} \le O(K \log K + K \log \log n + \log n)$ with probability exceeding $1 - \Omega(n^{-2})$.

\subsection{Proof of Proposition \ref{loglikelihood_Lipschitz_char}} \label{app:proof_loglikelihood_Lipschitz_char}
The log-likelihood $\LL(\matY; \vgamma)$ can be expressed as the sum 
$f(\vgamma) + g(\vgamma)$ with $f(\vgamma) = 
-L  \log{\lv \matSigma_{\vgamma} \rv}$ and $g(\vgamma) = -L \text{tr}\lb \matSigma_{\vgamma}^{-1} \matR_{\vecy \vecy} \rb$. Here, $\matSigma_{\vgamma} = \covy$. 

First, we derive an upper bound for the Lipschitz constant of 
$f(\vgamma) = -L \log|\matSigma_{\vgamma}|$ for $\vgamma \in \Theta(\setS)$. 
By the mean value theorem, any upper bound for $\Vert\nabla_{\vgamma} f(\vgamma)\Vert_2$ 
also serves as an upper bound for the Lipschitz constant of $f$. So, we 
derive an upper bound for $\Vert\nabla_{\vgamma} f(\vgamma)\Vert_{2}$. Note that 
$\lv \frac{\partial f(\vgamma)}{\partial \vgamma(i)} \rv = L(\veca_{i}^{T} \matSigma_{\vgamma}^{-1}\veca_{i})$ 
for $i \in \setS$, and $0$ otherwise. Here, $\veca_{i}$ denotes the $i^{\text{th}}$ 
column of $\matA$. Then, $\Vert\nabla_{\vgamma} f(\vgamma)\Vert_{2}$ can be 
upper bounded as shown below.
\begin{align}
\Vert\nabla_{\vgamma} f(\vgamma)\Vert_{2} &\le \Vert\nabla_{\vgamma} f(\vgamma)\Vert_{1}
\;=\; L \sum_{i \in \setS} \veca_{i}^{T} \matSigma_{\vgamma}^{-1}\veca_{i}
\nonumber \\
& \hspace{0cm} = L \lb \text{tr}\lb \matA_{\setS}^{T} (\sigma^{2}\matI_{m} + \matA_{\setS} \matGamma_{\setS} \matA_{\setS}^{T})^{-1} \matA_{\setS} \rb \rb
\nonumber \\
& \hspace{0cm} = L \lb \text{tr}\lb \matGamma_{\setS}^{-1/2} \tilde{\matA}_{\setS}^{T} (\sigma^{2}\matI_{m} + \tilde{\matA}_{\setS}  \tilde{\matA}_{\setS}^{T})^{-1} \tilde{\matA}_{\setS} 
\matGamma_{\setS}^{-1/2}\rb \rb
\nonumber \\
& \hspace{0cm} \stackrel{(a)}{=} L \lvvv \matGamma_{\setS}^{-1} \rvvv_{2} \text{tr}\lb  \tilde{\matA}_{\setS}^{T} (\sigma^{2}\matI_{m} + \tilde{\matA}_{\setS}  \tilde{\matA}_{\setS}^{T})^{-1} \tilde{\matA}_{\setS} \rb
\nonumber \\	
& \hspace{0cm} 
\; \stackrel{(b)}{\le} \; \lb L \min \lb m, |\setS| \rb \rb / \vgammamin. 
\label{log_likelihood_lipschitz_intr2}
\end{align}
where $\tilde{\matA}_{\setS} =  \matA_{\setS} \matGamma_{\setS}^{1/2}$. In the above, 
step ($a$) follows from the trace inequality $\text{tr}(\matA \matB) \le 
\Vert|\matA\Vert|_{2}\text{tr}(\matB)$ for any positive definite matrices $\matA$ and $\matB$. 
Step (b) follows from the observation that input argument of the trace operator has 
$\min{(m, |\setS|)}$ nonzero eigenvalues, all of them less than unity.

We now shift focus to the second term $g(\vgamma)$ of the loglikelihood.
Note that $\lv \frac{\partial g(\vgamma)}{\partial \vgamma(i)} \rv = $ $ L(\veca_{i}^{T} \matSigma^{-1}_{\vgamma} \matR_{Y} \matSigma_{\vgamma}^{-1}\veca_{i})$ 
for $i \in \setS$, and $0$ otherwise. Then, $\Vert\nabla_{\vgamma} g(\vgamma)\Vert_{2}$ 
can be upper bounded as
\begin{eqnarray}
\Vert\nabla g(\vgamma)\Vert_{2}
&\le& \Vert\nabla_{\vgamma} g(\vgamma)\Vert_{1}   
= L \sum_{i \in \setS} \veca_{i}^{T} \matSigma^{-1}_{\vgamma} \matR_{Y} \matSigma_{\vgamma}^{-1}\veca_{i} 
\nonumber \\
&=& L \lb \text{tr} \lb \matA_{\setS}^{T} \matSigma^{-1}_{\vgamma} \matR_{Y} \matSigma_{\vgamma}^{-1} \matA_{\setS}
\rb \rb \nonumber  \\
&\le& L \lvvv \matR_{\matY} \rvvv_{2} \text{tr} \lb \matA_{\setS}^{T} \matSigma^{-1}_{\vgamma} \matSigma_{\vgamma}^{-1} \matA_{\setS} \rb 
\nonumber  \\
&=& L \lvvv \matR_{\matY} \rvvv_{2}
\text{tr} \lb \matGamma_{\setS}^{-1/2}\tilde{\matA}_{\setS}^{T} \matSigma^{-1}_{\vgamma} \matSigma_{\vgamma}^{-1} \tilde{\matA}_{\setS} \matGamma_{\setS}^{-1/2} 
\rb \nonumber \\
&\le& \lb L \lvvv \matR_{\matY} \rvvv_{2} / \vgammamin\rb 
\text{tr} \lb \tilde{\matA}_{\setS}^{T} \matSigma^{-1}_{\vgamma} \matSigma_{\vgamma}^{-1} \tilde{\matA}_{\setS} \rb
\nonumber \\
&\le& \lb L \lvvv \matR_{\matY} \rvvv_{2} \lvvv \matSigma_{\vgamma}^{-1} \rvvv_{2} 
/ \vgammamin \rb 
\text{tr} \lb \tilde{\matA}_{\setS}^{T} \matSigma_{\vgamma}^{-1} \tilde{\matA}_{\setS} \!\rb
\nonumber \\
&\stackrel{(a)}{\le}& \lb L \lvvv \matR_{\matY} \rvvv_{2} \lvvv \matSigma_{\vgamma}^{-1} \rvvv_{2} 
\min{(m, |\setS|)} \rb/ \vgammamin
\label{eqn_lipschitz_char_intr1} 
\nonumber \\
&\stackrel{(b)}{\le}& \lb L \lvvv \matR_{\matY} \rvvv_{2}  
\min{(m, |\setS|)} \rb / \vgammamin \sigma^2 . 
\label{eqn_lipschitz_char_intr2} 
\label{log_likelihood_lipschitz_intr3}
\end{eqnarray}
where $\tilde{\matA}_{\setS} \triangleq \matA_{\setS} \matGamma_{\setS}^{1/2}$. The inequality in (\ref{eqn_lipschitz_char_intr1}$a$) follows from $\lb \tilde{\matA}_{\setS}^{T} \matSigma_{\vgamma}^{-1} \tilde{\matA}_{\setS} \rb$ having $\min{(m, |\setS|)}$ nonzero eigenvalues, all of them less than unity. The last inequality in (\ref{eqn_lipschitz_char_intr2}$b$) is due to $\lvvv \matSigma_{\vgamma}^{-1}\rvvv_{2} \le 1/\sigma^2$. Finally, the Lipschitz constant $C_{\LL, \setS}$ can be bounded as $C_{\LL, \setS} \le 
\Vert\nabla_{\vgamma} f(\vgamma)\Vert_{2} + \Vert\nabla_{\vgamma} g(\vgamma)\Vert_{2}$. Thus, 
by combining \eqref{log_likelihood_lipschitz_intr2} and \eqref{log_likelihood_lipschitz_intr3}, \textcolor{black}{and noting that $\min (m, |\mathcal{S}|) \le m$,} we obtain the desired result. 


\subsection{Proof of Proposition \ref{renyi_div_lb_gaussian_case_form2}} \label{app:proof_renyi_div_lb_gaussian_case_form2}
We assume that $\text{supp}(\vgamma_{1}) \backslash 
\text{supp}(\vgamma_{2}) \neq \phi$. The proof also holds for the case where $\text{supp}(\vgamma_{1}) \subset \text{supp}(\vgamma_{2})$ by swapping $\vgamma_{1}$ with $\vgamma_{2}$ and invoking the symmetry of $\DD_{1/2}(p_{\vgamma_{1}}, p_{\vgamma_{2}})$ with respect to its input arguments.    

Let $\mu^{*}$ be the largest eigenvalue of $\matSigma_{\vgamma_{1}}^{\frac{1}{2}} \matSigma_{\vgamma_{2}}^{-1} \matSigma_{\vgamma_{1}}^{\frac{1}{2}} $. Then, 
\begin{eqnarray}
\mu^{*} &\ge&  
 \frac{\text{tr}\lb \matSigma_{\vgamma_{2}}^{-1} \matSigma_{\vgamma_{1}} \rb}{m} 
 =
 \frac{1}{m} \ls \sigma^{2} \text{tr}\lb \matSigma_{\vgamma_{2}}^{-1} \rb + \text{tr}\lb \matSigma_{\vgamma_{2}}^{-1} \matA \matGamma_{1}\matA^{T} \rb \rs
 \nonumber \\
 & \ge &   
 \frac{1}{m} \text{tr}\lb \matSigma_{\vgamma_{2}}^{-1} \matA \matGamma_{1}\matA^{T} \rb. 
 \label{proof_disc_mat_lb_eq1}
\end{eqnarray}
Here, the second step is setting $\matSigma_{\vgamma_{1}}=\sigma^{2} \matI_{m} + \matA \matGamma_{1} \matA^{T}$.
\textcolor{black}{The last inequality is obtained by dropping the strictly positive 
$\frac{\sigma^{2}}{m} \text{tr}\lb \matSigma_{\vgamma_{2}}^{-1} \rb$ term.} 

Let $\setS_{1}$ and $\setS_{2}$ be the nonzero supports of $\vgamma_{1}$ and 
$\vgamma_{2}$, respectively. Further, let the eigendecomposition of $\matSigma_{\vgamma_{2}}$ be $\matU \Lambda \matU^{T}$, where $\Lambda = \diag{\lb \lambda_{1}, \ldots, \lambda_{m} \rb}$,  $\lambda_{i}$'s 
are the eigenvalues of $\matSigma_{\vgamma_{2}}$ 
and $\matU$ is a unitary matrix with columns as the eigenvectors of $\matSigma_{\vgamma_2}$. 
Then, $\matU$ can be 
partitioned as $\ls \matU_{{2}} \; \matU_{{2}^{\perp}} \rs$, where the columns 
of $\matU_{{2}}$ and $\matU_{{2}^{\perp}}$ span the orthogonal complementary 
subspaces $\text{Col}(\matA_{\setS_{2}})$ and $\text{Col}(\matA_{\setS_{2}})^{\perp}$, respectively. 
Further, let $\Lambda_{{2}}$ and $\Lambda_{{2}^{\perp}}$ be $|\setS_{2}|\times |\setS_{2}|$ and 
$((m-|\setS_{2}|)\times (m - |\setS_{2}|))$ sized diagonal matrices containing the eigenvalues in 
$\Lambda$ corresponding to the eigenvectors in $\matU_2$ and $\matU_{2^\perp}$, respectively. 
We observe that $\Lambda_{2^{\perp}} = \sigma^2 \matI_{m - |\setS_{2}|}$. 

By setting $\matSigma_{\vgamma_{2}}^{-1} = \matU_{{2}} \Lambda_{{2}}^{-1} \matU_{{2}}^{T} + 
\matU_{{2}^{\perp}} \Lambda_{{2}^{\perp}}^{-1} \matU_{{2}^{\perp}}^{T}$ in 
\eqref{proof_disc_mat_lb_eq1}, we get
\begin{eqnarray}
  \mu^{*}  &\ge&  
 \frac{1}{m} \lb 
 \text{tr}\lb \matU_{{2}} \Lambda_{{2}}^{-1} \matU_{{2}}^{T} \matA \matGamma_{1}\matA^{T} \rb 
  +
 \text{tr}\lb \matU_{{2}^{\perp}} \Lambda_{{2}^{\perp}}^{-1} \matU_{{2}^{\perp}}^{T} 
  \matA \matGamma_{1}\matA^{T} \rb
 \rb 
 \nonumber \\
&\ge&  
 \frac{1}{m} 
\text{tr}\lb \Lambda_{{2}^{\perp}}^{-1} \matU_{{2}^{\perp}}^{T} \matA \matGamma_{1}\matA^{T} \matU_{{2}^{\perp}} \rb, 
\nonumber
\end{eqnarray}
where the last inequality is due to nonnegativity of the first term.
Since $\matU_{{2}^{\perp}}^{T} \matA_{\setS_{2}} = 0$ by construction of $\matU_{{2}^{\perp}}$, 
\begin{eqnarray}
 \mu^{*} &\ge &
 \frac{1}{m} 
\text{tr}\lb \Lambda_{{2}^{\perp}}^{-1} \matU_{{2}^{\perp}}^{T} \matA_{\setS_{2}^{c}} \matGamma_{1, \setS_{2}^{c}}\matA_{\setS_{2}^{c}}^{T} \matU_{{2}^{\perp}} \rb 
\nonumber \\
&=&  
 \frac{1}{m \sigma^2} 
 \sum_{i = 1}^{m - |\setS_{2}|}  
 (\vecu_{2^{\perp},i})^{T} \matA_{\setS_{2}^{c}} \matGamma_{1, \setS_{2}^{c}} \matA_{\setS_{2}^{c}}^{T} 
 \vecu_{2^{\perp},i} 
\nonumber \\
 &=& 
 \frac{1}{m \sigma^2} 
 \sum_{i =1}^{m - |\setS_{2}|} 
 (\vecu_{2^{\perp},i})^{T} \matA_{\setS_{1} \backslash \setS_{2}} \matGamma_{1, \setS_{1} \backslash \setS_{2}}\matA_{\setS_{1} \backslash \setS_{2}}^{T} \vecu_{2^{\perp},i}. 
\label{proof_disc_mat_lb_eq2}
\end{eqnarray}

In the above, $\vecu_{2^{\perp},i}$ denotes the $i^{\text{th}}$ column of $\matU_{2^{\perp}}$.
The last equality is obtained by observing that the nonzero elements of $\vgamma_{1, \setS_{2}^{c}}$ 
are located in the index set $\setS_{1} \backslash \setS_{2}$. 

We now prove that if $K < \text{spark}(\matA) -1$, then there exists at least 
one strictly positive term in the above summation. Let us assume the contrary, i.e., let each term in the summation in \eqref{proof_disc_mat_lb_eq2} 
be equal to zero. This implies that the columns of $\matU_{2^{\perp}}$ belong to 
$\text{Null}(\matGamma^{1/2}_{1, \setS_{1} \backslash \setS_{2}} \matA^{T}_{\setS_{1} \backslash \setS_{2}} )$, 
which means that they also belong to $\text{Null}(\matA_{\setS_{1} \backslash \setS_{2}} \matGamma_{1, \setS_{1} \backslash \setS_{2}} \matA^{T}_{\setS_{1} \backslash \setS_{2}} )$. 
Since, for a symmetric matrix, the row and column spaces are equal and orthogonal to 
the null space of the matrix, it follows that 
$\text{Col}(\matA_{\setS_{1} \backslash \setS_{2}} \matGamma_{1, \setS_{1} \backslash \setS_{2}} \matA^{T}_{\setS_{1} \backslash \setS_{2}} )$ 
(same as $\text{Col}(\matA_{\setS_{1} \backslash \setS_{2}} \matGamma_{1, \setS_{1} \backslash \setS_{2}}^{1/2} )$)  
is spanned by the columns of $\matU_2$, or equivalently by columns of $\matA_{\setS_{2}}$. 
Thus, every column in $\matA_{\setS_{1} \backslash \setS_{2}}$ can be expressed as a linear combination of columns 
in $\matA_{\setS_{2}}$. Since $|\setS_{2}| \le K$, this contradicts our initial assumption that 
$K + 1 < \text{spark}(\matA)$.
Therefore, we conclude that there is at least one strictly positive term in the summation in~\eqref{proof_disc_mat_lb_eq2}, and consequently 
there exists a constant $c_{1}>0$ such that $ \mu^{*} \ge c_{1}/\sigma^{2}$.

\ifdefined \SKIPTEMP
\subsection{???}
\textcolor{black}{
\begin{lemma}
Let $\matA$ be an $m \times m$ sized symmetric real-valued (need not be positive semi-definite) matrix. If there exist $m$ linearly independent vectors 
in $\Real^{m}$, namely $\vecx_{1}, \vecx_{2}, \ldots, \vecx_{m}$ such that $\vecx_{i}^{T} \matA \vecx_{i} = 0$ 
for $i = 1$ to $m$, then it implies that $\matA = \mathbf{0}_{m \times m}$.
\end{lemma}
\begin{proof}
We show that $\matA = \mathbf{0}_{m \times m}$ by proving that the column-space of $\matA$, denoted by $\text{Col}(\matA)$, is trivial and contains only the zero vector, under the existence of the said linearly independent vectors. 
\end{proof}
}
\fi


\bibliographystyle{IEEEtran}
\bibliography{IEEEabrv,bibJournalList,MSBL_supp_recovery}

\end{document}